\documentclass{article} 
\usepackage[margin=3.9cm]{geometry}
\usepackage[utf8]{inputenc} 
\usepackage{fixltx2e} 
\usepackage{breakurl} 
\usepackage{mathtools}
\usepackage{hyperref}
\usepackage{amsthm}
\usepackage{complexity}
\usepackage{balance}
\usepackage{mathptmx} 
\usepackage{mathdots}
\RequirePackage{latexsym}
\RequirePackage{mathtools}
\RequirePackage{amsfonts}
\usepackage{amssymb}
\RequirePackage{array}
\RequirePackage{xspace}
\RequirePackage{cite}
\usepackage{url}

\newcommand{\val}[2]{\texttt{val}^{#1}(#2)}
\newcommand{\suchthat}{\;\ifnum\currentgrouptype=16 \middle\fi|\;}

\newcommand{\ie}{i.e., }
\newcommand{\eg}{e.g., }
\newcommand{\etal}{\textit{et al.}\xspace}
\newcommand{\wrt}{w.r.t.\ }

\usepackage[textsize=small]{todonotes}

\newcommand{\Q}{\mathbf{Q}\xspace}
\newcommand{\N}{\mathbf{N}\xspace}
\newcommand{\Z}{\mathbf{Z}\xspace}

\newcommand{\MPG}{MPG\xspace}
\newcommand{\EG}{EG\xspace}
\newcommand{\MCP}{MCP\xspace}

\def\C{{\cal C}}

\def\F{{\cal F}}
\def\W{{\cal W}}

\newcommand{\figref}[1]{Fig.~\ref{#1}}

\newtheorem{Thm}{Theorem}
\newtheorem{Lem}{Lemma}
\newtheorem{Prop}{Proposition}

\newtheorem{Rem}{Remark}

\usepackage[ruled,vlined,linesnumbered]{algorithm2e}
\usepackage{subfig}
\SetCommentSty{textsf}
\SetKwRepeat{DoWhile}{do}{while}
\SetAlFnt{\small} 
\makeatletter
\newcommand{\removelatexerror}{\let\@latex@error\@gobble}
\makeatother

\let\oldnl\nl
\newcommand{\nonl}{\renewcommand{\nl}{\let\nl\oldnl}}

\usepackage{tikz}
\usepackage{pgfplots}
\usetikzlibrary{calc,positioning,fit}
\usetikzlibrary{shapes,shapes.multipart,shapes.arrows}
\usetikzlibrary{decorations,decorations.pathmorphing,decorations.pathreplacing,decorations.markings,decorations.shapes}
\usetikzlibrary{arrows}
\usetikzlibrary{fit,backgrounds}
\usetikzlibrary{plotmarks}
\tikzstyle{node}=[circle,draw,inner sep=2pt,transform shape,minimum size=1.75em]
\tikzstyle{every picture}=[>=latex]
\tikzstyle{every label}=[inner sep=2pt]

\newcount\recurdepth
\begin{document}

\title{An Improved Pseudo-Polynomial Upper Bound for the 
			Value Problem and Optimal Strategy Synthesis \\ in Mean Payoff Games}

\author{Carlo Comin~\footnote{Department of Mathematics, University of Trento, Trento, Italy; 
 LIGM, Universit{\'e} Paris-Est, Marne-la-Vall{\'e}e, Paris, France. (carlo.comin@unitn.it)} \and  
Romeo Rizzi~\footnote{Department of Computer Science, University of Verona, Verona, Italy. (romeo.rizzi@univr.it)}}

\date{}
\maketitle 

\begin{abstract}
In this work we offer an $O(|V|^2 |E|\, W)$ pseudo-polynomial time 
deterministic algorithm for solving the Value Problem and Optimal Strategy Synthesis 
in Mean Payoff Games. This improves by a factor $\log(|V|\, W)$ 
the best previously known pseudo-polynomial time upper bound due to Brim,~\etal
The improvement hinges on a suitable characterization of values, 
and a description of optimal positional strategies, 
in terms of reweighted Energy Games and Small Energy-Progress Measures.
\end{abstract}


\section{Introduction}\label{sect:introduction}

A \emph{Mean Payoff Game} (\MPG) is a two-player infinite game $\Gamma := (V, E, w, \langle V_0, V_1 \rangle)$, 
which is played on a finite weighted directed graph, denoted $G^{\Gamma} := ( V, E, w )$,   
the vertices of which are partitioned into two classes, $V_0$ and $V_1$, 
according to the player to which they belong. 
It is assumed that $G^{\Gamma}$ has no sink vertex and that the weights of the arcs are integers, 
\ie $w:E\rightarrow \{-W, \ldots,  0, \ldots,  W\}$ for some $W\in\N$.

At the beginning of the game a pebble is placed on some vertex $v_s\in V$,  
and then the two players, named Player~0 and Player~1, 
move the pebble ad infinitum along the arcs. Assuming the pebble is currently on Player~0's vertex $v$, 
then he chooses an arc $(v,v')\in E$ going out of $v$ and moves the pebble to the destination vertex $v'$.
Similarly, assuming the pebble is currently on Player~1's vertex, 
then it is her turn to choose an outgoing arc.
The infinite sequence $v_s,v,v'\ldots$ of all the encountered vertices is a \emph{play}.
In order to play well, Player~0 wants to maximize 
the limit inferior of the long-run average weight of the traversed arcs, 
\ie to maximize $\liminf_{n\rightarrow\infty} \frac{1}{n}\sum_{i=0}^{n-1}w(v_i, v_{i+1})$, 
whereas Player~1 wants to minimize the $\limsup_{n\rightarrow\infty} \frac{1}{n}\sum_{i=0}^{n-1}w(v_i, v_{i+1})$.
Ehrenfeucht and Mycielski~\cite{EhrenfeuchtMycielski:1979} 
proved that each vertex $v$ admits a \emph{value}, denoted  
$\val{\Gamma}{v}$, which each player can secure by means of a \emph{memoryless} (or \emph{positional}) strategy, 
\ie a strategy that depends only on the current vertex position and not on the previous choices.

Solving an \MPG consists in computing the values of all vertices (\emph{Value Problem})   
and, for each player, a positional strategy that secures such values to that player (\emph{Optimal Strategy Synthesis}).
The corresponding decision problem lies in $\NP\cap \coNP$~\cite{ZwickPaterson:1996} 
and it was later shown by Jurdzi\'nski~\cite{Jurdzinski1998} to be recognizable with \emph{unambiguous} polynomial time 
non-deterministic Turing Machines, 
thus falling within the $\UP\cap\coUP$ complexity class.

The problem of devising efficient algorithms for solving \MPG{s} has been studied extensively in the literature.
The first milestone was that of Gurvich, Karzanov and Khachiyan~\cite{GKK88},   
in which they offered an \emph{exponential} time algorithm for solving a slightly wider class of \MPG{s} called \emph{Cyclic Games}.
Afterwards, Zwick and Paterson~\cite{ZwickPaterson:1996} devised the first deterministic procedure for computing values in \MPG{s}, 
and optimal strategies securing them, within a \emph{pseudo-polynomial} time and polynomial space. 
In particular, Zwick and Paterson established an $O(|V|^3 |E|\, W)$ upper bound for the time complexity of 
the Value Problem, as well as an upper bound of $O(|V|^4|E|\, W\,\log(|E|/|V|))$ 
for that of Optimal Strategy Synthesis~\cite{ZwickPaterson:1996}.

Recently, several research efforts have been spent in studying quantitative extensions 
of infinite games for modeling quantitative aspects of reactive systems~\cite{Chakrabarti03,Bouyer08,brim2011faster}. 
In this context quantities may represent, for example, 
the power usage of an embedded component, or the buffer size of a networking element.
These studies have brought to light interesting connections with \MPG{s}. 
Remarkably, they have recently led to the design of faster procedures for solving them.
In particular, Brim,~\etal~\cite{brim2011faster} devised faster deterministic algorithms 
for solving the Value Problem and Optimal Strategy Synthesis in \MPG{s} 
within $O(|V|^2 |E|\, W\log(|V|\,W))$ pseudo-polynomial time and polynomial space. 

To the best of our knowledge, this is the tightest pseudo-polynomial upper bound on the 
time complexity of \MPG{s} which is currently known.

Indeed, a wide spectrum of different approaches have been investigated in the literature.
For instance, Andersson and Vorobyov~\cite{Andersson06fastalgorithms} provided a fast 
\emph{sub-exponential} time \emph{randomized} algorithm for solving \MPG{s}, 
whose time complexity can be bounded as 
$O(|V|^2 |E|\, 
\exp(2\, \sqrt{|V|\, \ln(|E| / \sqrt{|V|})}+O(\sqrt{|V|}+\ln|E|)))$.
Furthermore, Lifshits and Pavlov~\cite{LifshitsPavlov:2007} devised 
an $O(2^{|V|}\, |V|\, |E|\, \log W)$ \emph{singly-exponential} time 
deterministic procedure by considering the \emph{potential theory} of \MPG{s}.

These results are summarized in Table~\ref{Table:Algorithms}.

\begin{table}[!htb]
\caption{Complexity of the main Algorithms for solving \MPG{s}.}
\label{Table:Algorithms}
\centering
\bgroup
\def\arraystretch{1.4}
\begin{tabular}{| c  c  c  c  | }
\hline
\vtop{\hbox{\strut }\hbox{\strut Algorithm }\hbox{\strut }}&\vtop{\hbox{\strut }\hbox{\strut Value Problem }\hbox{\strut }}  & 
\vtop{\hbox{\strut Optimal }\hbox{\strut Strategy } \hbox{\strut Synthesis}} & 
\vtop{\hbox{\strut }\hbox{\strut Note }\hbox{\strut }} \\
\hline
This work &  $O(|V|^2 |E|\, W)$  & $O(|V|^2 |E|\, W)$                      & Determ. \\ 
\cite{brim2011faster} & $O(|V|^2 |E|\, W\,\log(|V|\, W))$ & $O(|V|^2 |E|\, W\, \log(|V|\, W))$   &  Determ. \\
\cite{ZwickPaterson:1996} &  $\Theta(|V|^3 |E|\, W)$        & $\Theta(|V|^4 |E|\, W\, \log\frac{|E|}{|V|})$ &  Determ. \\
\cite{LifshitsPavlov:2007} &  $O(2^{|V|}\, |V|\, |E|\, \log W)$ &  n/a &  Determ. \\
\cite{Andersson06fastalgorithms} &  $O\Big(|V|^2 |E|\, 
e^{2\,\sqrt{|V|\, \ln\big(\frac{|E|}{\sqrt{|V|}}\big)}+O(\sqrt{|V|}+\ln|E|)}\Big)$ & same complexity & Random. \\
\hline
\end{tabular}
\egroup
\end{table}

\paragraph*{Contribution.}
The main contribution of this work is that to provide an $O(|V|^2|E|\,W)$ 
\emph{pseudo-polynomial} time and $O(|V|)$ space deterministic algorithm for 
solving the Value Problem and Optimal Strategy Synthesis in \MPG{s}. 
As already mentioned in the introduction, the best previously known procedure 
has a deterministic time complexity of $O(|V|^2 |E|\, W\log(|V|\,W))$, 
which is due to Brim,~\etal~\cite{brim2011faster}. In this way we improve the best 
previously known pseudo-polynomial time upper bound by a factor $\log(|V|\, W)$.
This result is summarized in the following theorem.

\begin{Thm}\label{Thm:algorithm}
There exists a deterministic algorithm for 
solving the Value Problem and Optimal Strategy Synthesis of \MPG{s} 
within $O(|V|^2 |E|\, W)$ time and $O(|V|)$ space, 
on any input \MPG $\Gamma=(V, E, w, \langle V_0, V_1 \rangle)$. 
Here, $W=\max_{e\in E}|w_e|$. 
\end{Thm}

In order to prove Theorem~\ref{Thm:algorithm}, 
this work points out a novel and suitable characterization of values,  
and a description of optimal positional strategies, in terms of certain reweighting operations 
that we will introduce later on in Section~\ref{sect:background}.  

In particular, 
we will show that the optimal value $\val{\Gamma}{v}$ 
of any vertex $v$ is the unique rational number $\nu$ for which 
$v$ \emph{``transits"} from the winning region of Player~0 to that of Player~1, 
with respect to reweightings of the form $w-\nu$.
This intuition will be clarified later on in Section~\ref{section:values}, 
where Theorem~\ref{Thm:transition_opt_values} is formally proved.

Concerning strategies, we will show that an optimal positional strategy  
for each vertex $v\in V_0$ is given by any arc $(v,v')\in E$ which 
is compatible with certain \emph{Small Energy-Progress Measures} (SEPMs) of the above mentioned reweighted arenas.
This fact is formally proved in Theorem~\ref{Thm:pos_opt_strategy} of Section~\ref{section:values}.

These novel observations are smooth, simple, and their proofs rely on elementary arguments.
We believe that they contribute to clarifying the interesting relationship between values, 
optimal strategies and reweighting operations 
(with respect to some previous literature, see e.g.~\cite{brim2011faster, LifshitsPavlov:2007}). 
Indeed, they will allow us to prove Theorem~\ref{Thm:algorithm}.

\paragraph*{Organization.}
This manuscript is organized as follows. 
In Section~\ref{sect:background}, we introduce some notation 
and provide the required background on infinite two-player games and related algorithmic results.
In Section~\ref{section:values}, a suitable relation between values, optimal strategies, and 
certain reweighting operations is investigated. 
In Section~\ref{sect:algorithm}, an $O(|V|^2 |E|\, W)$ pseudo-polynomial time and $O(|V|)$ space algorithm,  
for solving the Value Problem and Optimal Strategies Synthesis in \MPG{s}, 
is designed and analyzed by relying on the results presented in Section~\ref{section:values}.
In this manner, Section~\ref{sect:algorithm} actually provides 
a proof of Theorem~\ref{Thm:algorithm} which is our main result in this work.

\section{Notation and Preliminaries}\label{sect:background}
We denote by $\N$, $\Z$, $\Q$ the set of natural, integer, and rational numbers (respectively). 
It will be sufficient to consider integral intervals, \eg $[a,b]:=\{z\in\Z\mid a\leq z\leq b\}$ 
and $[a,b):=\{z\in\Z\mid a \leq z < b\}$ for any $a,b\in \Z$.  

\paragraph*{Weighted Graphs.}\; 
Our graphs are directed and weighted on the arcs. 
Thus, if $G=(V, E, w)$ is a graph, then every arc $e\in E$ is a triplet $e=(u,v,w_e)$, 
where $w_e = w(u,v)\in\Z$ is the \textit{weight} of $e$. 
The maximum absolute weight is $W := \max_{e\in E} |w_e|$.
Given a vertex $u\in V$, the set of its successors is  
$\texttt{post}(u) = \{ v \in V \mid (u,v) \in E \}$, whereas the set of its predecessors is 
$\texttt{pre}(u) = \{ v \in V \mid (v,u) \in E \}$. A \emph{path} is a sequence of vertices 
$v_0v_1\ldots v_n\ldots$ such that $(v_i, v_{i+1}) \in E$ for every $i$. 
We denote by $V^*$ the set of all (possibly empty) finite paths.
A \emph{simple path} is a finite path $v_0v_1\ldots v_n$ having no repetitions, 
\ie for any $i,j\in [0,n]$ it holds $v_i \neq v_j$ whenever $i\neq j$.
The \emph{length} of a simple path $\rho=v_0v_1\ldots v_n$ equals $n$ and it is denoted by $|\rho|$.
A \emph{cycle} is a path $v_0v_1\ldots v_{n-1}v_n$ such that $v_0\ldots v_{n-1}$ is simple and $v_n = v_0$.
The \emph{length} of a cycle $C=v_0v_1\ldots v_n$ equals $n$ and it is denoted by $|C|$.  
The \emph{average weight} of a cycle $v_0\ldots v_n$ is $\frac{1}{n} \sum_{i=0}^{n-1} w(v_i,v_{i+1})$. 
A cycle $C=v_0v_1\ldots v_n$ is \emph{reachable} from $v$ in $G$ 
if there exists a simple path $p=vu_1\ldots u_m$ in $G$ such that $p\cap C\neq\emptyset$. 

\paragraph*{Arenas.}
An \emph{arena} is a tuple $\Gamma = ( V, E, w, \langle V_0, V_1\rangle )$ 
where $G^{\Gamma} := ( V, E, w )$ is a finite weighted directed graph and $( V_0, V_1)$ is a partition 
of $V$ into the set $V_0$ of vertices owned by Player~0, and the set $V_1$ of vertices owned by Player~$1$.
It is assumed that $G^{\Gamma}$ has no sink, \ie $\texttt{post}(v)\neq\emptyset$ for every $v \in V$;
still, we remark that $G^{\Gamma}$ is not required to be a bipartite graph on colour classes $V_0$ and $V_1$. 
\figref{fig:arena_ex} depicts an example.

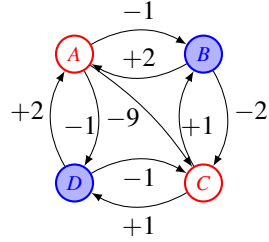
\begin{figure}[!h]
\center
\begin{tikzpicture}[arrows=->,scale=0.8,node distance=1.5 and 1.5]
 		\node[node, thick, color=red] (A) {$A$};
		\node[node, thick, color=blue, fill=blue!30, right=of A] (B) {$B$};
		\node[node, thick, color=red, below=of B] (C) {$C$};
		\node[node, thick, color=blue, fill=blue!30, left=of C] (D) {$D$};	
		\draw[] (A) to [bend left=30] node[above] {$-1$} (B); 
		\draw[] (B) to [bend left=30] node[above, xshift=2ex, yshift=-1ex] {$-2$} (C);
		\draw[] (C) to [bend left=30] node[below] {$+1$} (D);
		\draw[] (D) to [bend left=30] node[above, xshift=-2ex, yshift=-1ex] {$+2$} (A);
		\draw[] (A) to [bend left=30] node[below, xshift=-1.5ex, yshift=1ex] {$-1$} (D); 
		\draw[] (D) to [bend left=30] node[below] {$-1$} (C);
		\draw[] (C) to [bend left=30] node[below, xshift=1.5ex, yshift=1ex] {$+1$} (B);
		\draw[] (B) to [bend left=30] node[above] {$+2$} (A);
		\draw[] (A) to [bend left=10] node[below, xshift=-2ex, yshift=1ex] {$-9$} (C);
\end{tikzpicture}
\caption{An arena $\Gamma$.}\label{fig:arena_ex}
\end{figure}

A game on $\Gamma$ is played for infinitely many rounds by two players moving a pebble along the arcs 
of $G^{\Gamma}$. At the beginning of the game we find the pebble on some vertex 
$v_s\in V$, which is called the \emph{starting position} of the game.
At each turn, assuming the pebble is currently on a vertex $v\in V_i$ (for $i=0, 1$), 
Player~$i$ chooses an arc $( v,v')\in E$ and then the next turn starts with the pebble on $v'$.

A \emph{play} is any infinite path $v_0v_1\ldots v_n\ldots\in V^*$ in $\Gamma$. 
For any $i\in \{0,1\}$, a strategy of Player~$i$ is any function $\sigma_i:V^*\times V_i\rightarrow V$ 
such that for every finite path $p'v$ in $G^{\Gamma}$, where $p'\in V^*$ and $v\in V_i$, 
it holds that $(v, \sigma_i(p', v))\in E$. 
A strategy $\sigma_i$ of Player $i$ is \emph{positional} (or \emph{memoryless}) 
if $\sigma_i(p, v_n) = \sigma_i(p', v'_m)$ for every finite paths $pv_n=v_0\ldots v_{n-1}v_n$ 
and $p'v'_m=v'_0\ldots v'_{m-1}v'_m$ in $G^{\Gamma}$ such that $v_n=v'_m\in V_i$. 
The set of all the positional strategies of Player~$i$ is denoted by $\Sigma^M_i$. 
A play $v_0v_1\ldots v_n\ldots $ is \emph{consistent} with a strategy $\sigma\in\Sigma_i$ 
if $v_{j+1} = \sigma(v_0v_1\ldots v_j)$ whenever $v_j\in V_i$. 

Given a starting position $v_s\in V$, the \emph{outcome} of strategies $\sigma_0 \in\Sigma_0$ and 
$\sigma_1 \in\Sigma_1$, denoted $\texttt{outcome}^{\Gamma}(v_s, \sigma_0, \sigma_1)$, is the unique play  
that starts at $v_s$ and is consistent with both $\sigma_0$ and $\sigma_1$.

Given a memoryless strategy $\sigma_i\in\Sigma^M_i$ of Player~$i$ in $\Gamma$, then  
$G^{\Gamma}_{\sigma_i}=( V, E_{\sigma_i}, w)$ is the graph obtained from $G^{\Gamma}$ by removing all the 
arcs $( v,v')\in E$ such that $v\in V_i$ and $v'\neq \sigma_i(v)$; 
we say that $G^{\Gamma}_{\sigma_i}$ is obtained from $G^{\Gamma}$ \emph{by projection} \wrt $\sigma_i$. 

Concluding this subsection, the notion of \emph{reweighting} is recalled.
For any weight function $w,w':E\rightarrow \Z $, the \emph{reweighting} of 
$\Gamma=(V, E, w, \langle V_0, V_1\rangle )$ \wrt $w'$ 
is the arena $\Gamma^{w'}= (V, E, w', \langle V_0, V_1\rangle )$.
Also, for $w:E\rightarrow \Z$ and any $\nu\in \Z$, 
we denote by $w+\nu$ the weight function $w'$ defined as $w'_e := w_e+\nu$ for every $e\in E$.
Indeed, we shall consider reweighted games of the form $\Gamma^{w-q}$, for some $q\in \Q$. 
Notice that the corresponding weight function $w':E\rightarrow\Q:e\mapsto w_e-q$ is rational, 
while we required the weights of the arcs to be always integers. 
To overcome this issue, it is sufficient to re-define $\Gamma^{w-q}$ by  
scaling all the weights by a factor equal to the denominator of $q\in \Q$, namely, 
to re-define: $\Gamma^{w-q}:=\Gamma^{D\cdot w-N}$, where $N,D\in\N$ are such that $q=N/D$ and $\gcd(N,D)=1$. 
This re-scaling will not change the winning regions of the corresponding games, 
and it has the significant advantage of allowing for a discussion (and an algorithmics) 
which is strictly based on integer weights.

\paragraph*{Mean Payoff Games.}
A \emph{Mean Payoff Game}~(\MPG)~\cite{brim2011faster, ZwickPaterson:1996, EhrenfeuchtMycielski:1979} 
is a game played on some arena $\Gamma$ for infinitely many rounds by two opponents,  
Player~$0$ gains a payoff defined as the long-run average weight of the play, whereas Player~$1$ loses that value. 
Formally, the Player~$0$'s \emph{payoff} of a play $v_0v_1\ldots v_n\ldots $ in $\Gamma$ is defined as follows: 
\[\texttt{MP}_0(v_0v_1\ldots v_n\ldots):=\liminf_{n\rightarrow\infty} \frac{1}{n}\sum_{i=0}^{n-1}w(v_i, v_{i+1}).\]
The value \emph{secured} by a strategy $\sigma_0\in\Sigma_0$ in a vertex $v$ is defined as:
\[\texttt{val}^{\sigma_0}(v):=
\inf_{\sigma_1\in\Sigma_1}\texttt{MP}_0\big(\texttt{outcome}^{\Gamma}(v, \sigma_0, \sigma_1)\big),\]
Notice that payoffs and secured values can be defined symmetrically for the Player~$1$ 
(\ie by interchanging the symbol \emph{0} with \emph{1} and \emph{inf} with \emph{sup}).

Ehrenfeucht and Mycielski~\cite{EhrenfeuchtMycielski:1979} proved that each vertex $v\in V$ admits a unique \emph{value}, 
denoted $\val{\Gamma}{v}$, which each player can secure by means of a \emph{memoryless} (or \emph{positional}) strategy.
Moreover, \emph{uniform} positional optimal strategies do exist for both players, 
in the sense that for each player there exist at least one positional strategy which can be used to secure all the optimal values, 
independently with respect to the starting position $v_s$.
Thus, for every \MPG $\Gamma$, there exists a strategy $\sigma_0\in\Sigma^M_0$ such that 
$\val{\sigma_0}{v}\geq \val{\Gamma}{v}$ for every $v\in V$,  
and there exists a strategy $\sigma_1\in\Sigma^M_1$ such that $\val{\sigma_1}{v}\leq \val{\Gamma}{v}$ for every $v\in V$.
Indeed, the \emph{(optimal) value} of a vertex $v\in V$ in the \MPG $\Gamma$ is given by:
\[\val{\Gamma}{v}= 
\sup_{\sigma_0\in\Sigma_0}\val{\sigma_0}{v} = \inf_{\sigma_1\in\Sigma_1}\val{\sigma_1}{v}.\]
Thus, a strategy $\sigma_0\in\Sigma_0$ is \emph{optimal} if $\texttt{val}^{\sigma_0}(v)=\val{\Gamma}{v}$ for all $v\in V$.  
A strategy $\sigma_0\in\Sigma_0$ is said to be \emph{winning} for Player~$0$ if $\texttt{val}^{\sigma_0}(v)\geq 0$,  
and $\sigma_1\in\Sigma_1$ is winning for Player~$1$ if $\texttt{val}^{\sigma_1}(v) < 0$.
Correspondingly, a vertex $v\in V$ is a \emph{winning starting position} for Player~$0$ if $\val{\Gamma}{v}\geq 0$, 
otherwise it is winning for Player~$1$. 
The set of all winning starting positions of Player~$i$ is denoted by $\W_i$ for $i\in \{0,1\}$.
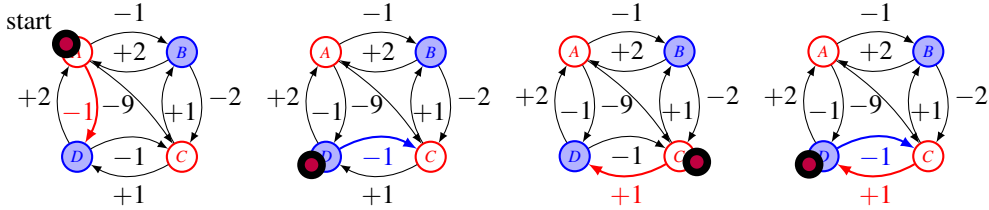
\begin{figure}[!h]
	\centering
	\begin{tikzpicture}[arrows=->,scale=.65,node distance=1.5 and 1.5]
 		\node[node, thick, color=red] (A) {$A$};
		\node[node, thick, color=blue, fill=blue!30, right=of A] (B) {$B$};
		\node[node, thick, color=red, below=of B] (C) {$C$};
		\node[node, thick, color=blue, fill=blue!30, left=of C] (D) {$D$};	
  \node[node, scale=0.7, xshift=-2ex, yshift=1.5ex, fill=purple, draw=black, line width=1mm, label={above left:start}] (pebble) {}; 
		\draw[] (A) to [bend left=30] node[above] {$-1$} (B); 
		\draw[] (B) to [bend left=30] node[above, xshift=2ex, yshift=-1ex] {$-2$} (C);
		\draw[] (C) to [bend left=30] node[below] {$+1$} (D);
		\draw[] (D) to [bend left=30] node[above, xshift=-2ex, yshift=-1ex] {$+2$} (A);
		\draw[thick, color=red] (A) to [ bend left=30] node[below, xshift=-1.5ex, yshift=1ex] {$-1$} (D); 
		\draw[] (D) to [bend left=30] node[below] {$-1$} (C);
		\draw[] (C) to [bend left=30] node[below, xshift=1.5ex, yshift=1ex] {$+1$} (B);
		\draw[] (B) to [bend left=30] node[above] {$+2$} (A);
		\draw[] (A) to [bend left=10] node[below, xshift=-1.5ex, yshift=1ex] {$-9$} (C);
	\end{tikzpicture}
	\begin{tikzpicture}[arrows=->,scale=.65,node distance=1.5 and 1.5]
 		\node[node, thick, color=red] (A) {$A$};
		\node[node, thick, color=blue, fill=blue!30, right=of A] (B) {$B$};
		\node[node, thick, color=red, below=of B] (C) {$C$};
		\node[node, thick, color=blue, fill=blue!30, left=of C] (D) {$D$};	
\node[node, left=of C, scale=0.7, xshift=-3.3ex, yshift=-1.5ex, fill=purple, draw=black, line width=1mm] (pebble) {};	
		\draw[] (A) to [bend left=30] node[above] {$-1$} (B); 
		\draw[] (B) to [bend left=30] node[above, xshift=2ex, yshift=-1ex] {$-2$} (C);
		\draw[] (C) to [bend left=30] node[below] {$+1$} (D);
		\draw[] (D) to [bend left=30] node[above, xshift=-2ex, yshift=-1ex] {$+2$} (A);
		\draw[] (A) to [bend left=30] node[below, xshift=-1.5ex, yshift=1ex] {$-1$} (D); 
		\draw[thick, color=blue] (D) to [bend left=30] node[below] {$-1$} (C);
		\draw[] (C) to [bend left=30] node[below, xshift=1.5ex, yshift=1ex] {$+1$} (B);
		\draw[] (B) to [bend left=30] node[above] {$+2$} (A);
		\draw[] (A) to [bend left=10] node[below, xshift=-1.5ex, yshift=1ex] {$-9$} (C);
	\end{tikzpicture}
	\begin{tikzpicture}[arrows=->,scale=.65,node distance=1.5 and 1.5]
 		\node[node, thick, color=red] (A) {$A$};
		\node[node, thick, color=blue, fill=blue!30, right=of A] (B) {$B$};
		\node[node, thick, color=red, below=of B] (C) {$C$};
		\node[node, thick, color=blue, fill=blue!30, left=of C] (D) {$D$};	
\node[node, below=of B, scale=0.7, xshift=3.2ex, yshift=-1.7ex, fill=purple, draw=black, line width=1mm] (pebble) {};
		\draw[] (A) to [bend left=30] node[above] {$-1$} (B); 
		\draw[] (B) to [bend left=30] node[above, xshift=2ex, yshift=-1ex] {$-2$} (C);
		\draw[thick, color=red] (C) to [bend left=30] node[below] {$+1$} (D);
		\draw[] (D) to [bend left=30] node[above, xshift=-2ex, yshift=-1ex] {$+2$} (A);
		\draw[] (A) to [bend left=30] node[below, xshift=-1.5ex, yshift=1ex] {$-1$} (D); 
		\draw[] (D) to [bend left=30] node[below] {$-1$} (C);
		\draw[] (C) to [bend left=30] node[below, xshift=1.5ex, yshift=1ex] {$+1$} (B);
		\draw[] (B) to [bend left=30] node[above] {$+2$} (A);
		\draw[] (A) to [bend left=10] node[below, xshift=-1.4ex, yshift=1ex] {$-9$} (C);
	\end{tikzpicture}
	\begin{tikzpicture}[arrows=->,scale=.65,node distance=1.5 and 1.5]
 		\node[node, thick, color=red] (A) {$A$};
		\node[node, thick,  color=blue, fill=blue!30, right=of A] (B) {$B$};
		\node[node, thick,  color=red, below=of B] (C) {$C$};
		\node[node, thick,  color=blue, fill=blue!30, left=of C] (D) {$D$};	
		\node[node, left=of C, scale=0.7, xshift=-3.3ex, yshift=-1.5ex, 
			fill=purple, draw=black, line width=1mm] (pebble) {};	
		\draw[] (A) to [bend left=30] node[above] {$-1$} (B); 
		\draw[] (B) to [bend left=30] node[above, xshift=2ex, yshift=-1ex] {$-2$} (C);
		\draw[thick, color=red] (C) to [bend left=30] node[below] {$+1$} (D);
		\draw[] (D) to [bend left=30] node[above, xshift=-2ex, yshift=-1ex] {$+2$} (A);
		\draw[] (A) to [bend left=30] node[below, xshift=-1.5ex, yshift=1ex] {$-1$} (D); 
		\draw[thick, color=blue] (D) to [bend left=30] node[below] {$-1$} (C);
		\draw[] (C) to [bend left=30] node[below, xshift=1.5ex, yshift=1ex] {$+1$} (B);
		\draw[] (B) to [bend left=30] node[above] {$+2$} (A);
		\draw[] (A) to [bend left=10] node[below, xshift=-1.4ex, yshift=1ex] {$-9$} (C);
	\end{tikzpicture}

\caption{An MPG on $\Gamma$, played from left to right, whose payoff equals $\frac{-1+1}{2}=0$.}
\end{figure}

A finite variant of \MPG{s} is well-known in the literature 
\cite{EhrenfeuchtMycielski:1979, ZwickPaterson:1996, brim2011faster}. 
Here, the game stops as soon as a cyclic sequence of vertices is traversed 
(\ie as soon as one of the two players moves the pebble into a previously visited vertex). 
It turns out that this finite variant is equivalent to the infinite one \cite{EhrenfeuchtMycielski:1979}.
Specifically, the values of an \MPG are in relationship with the average weights of its cycles, 
as stated in the next lemma.

\begin{Lem}[Brim,~\etal \cite{brim2011faster}]\label{Lem:reachable_cycle}
Let $\Gamma= ( V, E, w, \langle V_0, V_1\rangle )$ be an \MPG. 
For all $\nu\in\Q$, for all positional strategies $\sigma_0\in\Sigma^M_0$
of Player~$0$, and for all vertices $v\in V$, the value $\texttt{val}^{\sigma_0}(v)$ 
is greater than $\nu$ if and only if all cycles $C$ 
reachable from $v$ in the projection graph $G^{\Gamma}_{\sigma_0}$ 
have an average weight $w(C)/|C|$ greater than $\nu$.
\end{Lem} 
The proof of Lemma~\ref{Lem:reachable_cycle} follows from the memoryless determinacy of \MPG{s}. 
We remark that a proposition which is symmetric to Lemma~\ref{Lem:reachable_cycle} 
holds for Player~$1$ as well: for all $\nu\in\Q$, 
for all positional strategies $\sigma_1\in\Sigma^M_1$ of Player~$1$, and for all vertices $v\in V$, 
the value $\texttt{val}^{\sigma_1}(v)$ is less than $\nu$ if and only if 
all cycles reachable from $v$ in the projection graph $G^{\Gamma}_{\sigma_1}$
have an average weight less than $\nu$.

Also, it is well-known~\cite{brim2011faster, EhrenfeuchtMycielski:1979} 
that each value $\val{\Gamma}{v}$ is contained within the following set of rational numbers:
\[ S_{\Gamma}=\left\{ \frac{N}{D} \suchthat D\in [1, |V|],\, N\in [-D\, W, D\, W] \right\}.\]
Notice that $|S_{\Gamma}|\leq |V|^2 W$. 

\paragraph*{}
The present work tackles on the algorithmics of the following two classical problems:

\begin{itemize}
\item \emph{Value Problem.} Compute for each vertex $v\in V$ the (rational) optimal value $\val{\Gamma}{v}$.
\item \emph{Optimal Strategy Synthesis.} Compute an optimal positional strategy $\sigma_0\in\Sigma^M_0$. 
\end{itemize}

Currently, the asymptotically fastest pseudo-polynomial time algorithm  
for solving both problems is a deterministic procedure whose time complexity 
is $O(|V|^2|E|\, W\,\log (|V|\, W))$~\cite{brim2011faster}. 
This result has been achieved by devising a binary-search procedure that ultimately reduces the Value Problem 
and Optimal Strategy Synthesis to the resolution of yet another family of games known as the \emph{Energy Games}.
Even though we do not rely on binary-search in the present work, 
and thus we will introduce some truly novel ideas that diverge from the previous solutions, 
still, we will reduce to solving multiple instances of Energy Games. 
For this reason, the Energy Games are recalled in the next paragraph.

\paragraph*{Energy Games and Small Energy-Progress Measures.}
An \emph{Energy Game} (\EG) is a game that is played on an arena $\Gamma$ for infinitely many rounds by two opponents,
where the goal of Player~0 is to construct an infinite play $v_0v_1\ldots v_n\ldots$ such that for some initial 
\emph{credit} $c\in\N$ the following holds:
\begin{eqnarray}
	c + \sum_{i=0}^{j}w(v_i, v_{i+1})\geq 0\, \text{, for all } j \geq 0. 
\end{eqnarray}
Given a credit $c\in\N$, a play $v_0v_1\ldots v_n\ldots$ is \emph{winning} for Player~0 if it satisfies (1), 
otherwise it is winning for Player~1. 
A vertex $v\in V$ is a winning starting position for Player~0 if there exists an initial credit $c\in\N$ 
and a strategy $\sigma_0\in\Sigma_0$ such that, for every strategy $\sigma_1\in\Sigma_1$, 
the play $\texttt{outcome}^{\Gamma}(v, \sigma_0, \sigma_1)$ is winning for Player~0. 
As in the case of \MPG{s}, the \EG{s} are memoryless determined \cite{brim2011faster}, 
\ie for every $v\in V$, either $v$ is winning for Player~$0$ or $v$ is winning for Player~$1$, 
and (uniform) memoryless strategies are sufficient to win the game.
In fact, as shown in the next lemma, the decision problems of \MPG{s} and \EG{s} are intimately related.
\begin{Lem}[Brim,~\etal \cite{brim2011faster}]\label{Lem:relation_MPG_EG}
Let $\Gamma=( V, E, w, \langle V_0, V_1 \rangle )$ be an arena. For all threshold $\nu\in\Q$, 
for all vertices $v\in V$, Player~$0$ has a strategy in the \MPG $\Gamma$ that secures value 
at least $\nu$ from $v$ if and only if, for some initial credit $c\in\N$, 
Player~$0$ has a winning strategy from $v$ in the reweighted \EG $\Gamma^{w-\nu}$.
\end{Lem}

In this work we are especially interested in the \emph{Minimum Credit Problem} (\MCP) for \EG{s}: 
for each winning starting position $v$, compute the minimum initial credit $c^*=c^*(v)$ 
such that there exists a winning strategy $\sigma_0\in\Sigma^M_0$ for Player~$0$ starting from $v$.
A fast pseudo-polynomial time deterministic procedure for solving \MCP{s} comes from \cite{brim2011faster}.

\begin{Thm}[Brim,~\etal \cite{brim2011faster}]\label{Thm:VI}
There exists a deterministic algorithm for solving the MCP within 
$O(|V|\,|E|\,W)$ pseudo-polynomial time, 
on any input \EG $(V, E, w, \langle V_0, V_1\rangle)$. 
\end{Thm}
The algorithm mentioned in Theorem~\ref{Thm:VI} is the 
\emph{Value-Iteration} algorithm analyzed by Brim,~\etal in \cite{brim2011faster}. 
Its rationale relies on the notion of \emph{Small Energy-Progress Measures} (SEPM{s}).
These are bounded, non-negative and integer-valued functions that impose local conditions to ensure global properties on the
arena, in particular, witnessing that Player~0 has a way to enforce conservativity (\ie non-negativity of cycles) in the resulting game's graph.
Recovering standard notation, see e.g.~\cite{brim2011faster}, 
let us denote $\C_\Gamma=\{n\in\N\mid n\leq |V|\, W\}\cup\{\top\}$ and let $\preceq$ be the total order on 
$\C_\Gamma$ defined as $x\preceq y$ if and only if either $y=\top$ or $x,y\in\N$ and $x\leq y$.

In order to cast the minus operation to range over $\C_{\Gamma}$, 
let us consider an operator $\ominus:\C_\Gamma\times\Z\rightarrow \C_\Gamma$ defined as follows: 
\[
a\ominus b := \left\{ 
\begin{array}{ll}
	\max(0, a-b) & ,  \text{ if } a\neq \top \text{ and } a-b\leq |V|\,W ; \\
	a\ominus b = \top & , \text{ otherwise.} \\
\end{array}\right.
\]
Given an \EG $\Gamma$ on vertex set $V = V_0 \cup V_1$, a function $f:V\rightarrow\C_\Gamma$ is 
a \emph{Small Energy-Progress Measure} (SEPM) for $\Gamma$ if and only if the following two conditions are met: 
\begin{enumerate}
\item if $v\in V_0$, then $f(v)\succeq f(v') \ominus w(v,v')$ for \emph{some} $( v, v')\in E$;
\item if $v\in V_1$, then $f(v)\succeq f(v') \ominus w(v,v')$ for \emph{all} $( v, v')\in E$.
\end{enumerate}

The values of a SEPM, \ie the elements of the image $f(V)$,
are called the \emph{energy levels} of $f$.
It is worth to denote by $V_f=\{v\in V\mid f(v)\neq\top\}$ the set of vertices having finite energy.
Given a SEPM $f$ and a vertex $v\in V_0$, 
an arc $(v, v')\in E$ is said to be \emph{compatible with $f$} whenever $f(v)\succeq f(v')\ominus w(v,v')$;
moreover, a positional strategy $\sigma^f_0\in\Sigma^M_0$ is said to be \emph{compatible with $f$} whenever 
for all $v\in V_0$, if $\sigma^f_0(v)=v'$ then $( v,v' )\in E$ is compatible with $f$.
Notice that, as mentioned in \cite{brim2011faster}, if $f$ and $g$ are SEPM{s}, 
then so is the \emph{minimum function} defined as: $h(v)=\min\{f(v), g(v)\}$ for every $v\in V$.
This fact allows one to consider the \emph{least} SEPM,  
namely, the unique SEPM $f^*:V\rightarrow \C_\Gamma$ such that, 
for any other SEPM $g:V\rightarrow \C_\Gamma$, the following holds: $f^*(v)\preceq g(v)$ for every $v\in V$.
Also concerning SEPMs, we shall rely on the following lemmata.
The first one relates SEPMs to the winning region $\W_0$ of Player~0 in \EG{s}.
\begin{Lem}[Brim,~\etal \cite{brim2011faster}]\label{Lem:least_energy_prog_measure}
Let $\Gamma=( V, E, w, \langle V_0, V_1\rangle )$ be an \EG. 
\begin{enumerate}
\item If $f$ is any SEPM of the \EG $\Gamma$ and $v\in V_{f}$, 
then $v$ is a winning starting position for Player~$0$ in the \EG $\Gamma$.
Stated otherwise, $V_f\subseteq \W_0$;
\item If $f^*$ is the least SEPM of the \EG $\Gamma$, 
and $v$ is a winning starting position for Player~$0$ in the \EG $\Gamma$, then $v\in V_{f^*}$. 
Thus, $V_{f^*}=\W_0$.
\end{enumerate}
\end{Lem}
Also notice that the following bound holds on 
the energy levels of any SEPM (actually by definition of $\C_{\Gamma}$). 
\begin{Lem}\label{Lem:lepm_equals_mincredit}
Let $\Gamma=( V, E, w, \langle V_0, V_1 \rangle ) $ be an \EG. 
Let $f$ be any SEPM of $\Gamma$.
Then, for every $v\in V$ either $f(v)=\top$ or $0\leq f(v)\leq |V|\, W$.
\end{Lem}

\paragraph*{Value-Iteration Algorithm.} 
The algorithm devised by Brim,~\etal for solving the MCP in \EG{s} 
is known as \emph{Value-Iteration}~\cite{brim2011faster}.
Given an \EG $\Gamma$ as input, the Value-Iteration aims to compute the least SEPM $f^*$ of $\Gamma$.
This simple procedure basically relies 
on a lifting operator $\delta$.
Given $v\in V$, the \emph{lifting operator} 
$\delta(\cdot, v): [V\rightarrow \C_{\Gamma}]\rightarrow [V\rightarrow \C_\Gamma]$
is defined by $\delta(f,v)=g$, where:
\[
g(u)=\left\{
\begin{array}{ll}
f(u) & \text{ if } u\neq v \\
\min\{f(v') \ominus w(v,v') \mid v'\in \texttt{post}(v) \} & \text{ if } u=v\in V_0 \\
\max\{f(v') \ominus w(v,v') \mid v'\in \texttt{post}(v) \} & \text{ if } u=v\in V_1 \\
\end{array}
\right.
\]

We also need the following definition. Given a function $f:V\rightarrow \C_\Gamma$, we say that 
$f$ is \emph{inconsistent} in $v$ whenever one of the following two holds: 
\begin{enumerate}
\item $v\in V_0$ and \emph{for all} $v'\in \texttt{post}(v)$ it holds 
$f(v) \prec f(v') \ominus w(v,v')$;
\item  $v\in V_1$ and \emph{there exists} $v'\in \texttt{post}(v)$ such that  
$f(v) \prec f(v') \ominus w(v,v')$.
\end{enumerate}
 
To start with, the Value-Iteration algorithm initializes $f$ to the constant zero function, 
\ie $f(v)=0$ for every $v\in V$.
Furthermore, the procedure maintains a list $\L$ of vertices in order to witness the inconsistencies of $f$.
Initially, $v\in V_0\cap\L$ if and only if all arcs going out of $v$ are negative, while $v\in V_1\cap \L$ if and only 
if $v$ is the source of at least one negative arc. 
Notice that checking the above conditions takes time $O(|E|)$.

As long as the list $\L$ is nonempty, the algorithm picks a vertex $v$ from $\L$ and performs the following:
\begin{enumerate}
\item Apply the lifting operator $\delta(f,v)$ to $f$ in order to resolve the inconsistency of $f$ in $v$;
\item Insert into $\L$ all vertices $u\in\texttt{pre}(v)\setminus \L$ witnessing a new inconsistency due to the increase of $f(v)$.

(The same vertex can’t occur twice in $\L$, \ie there are no duplicate vertices in $\L$.)
\end{enumerate}
The algorithm terminates when $\L$ is empty. This concludes the description of the Value-Iteration algorithm.

As shown in~\cite{brim2011faster}, 
the update of $\L$ following an application of the lifting operator $\delta(f,v)$ requires $O(|\texttt{pre}(v)|)$ time. 
Moreover, a single application of the lifting operator $\delta(\cdot, v)$ takes $O(|\text{post}(v)|)$ time at most.
This implies that the algorithm can be implemented so that it will always halt within $O(|V|\,|E|\, W)$ time  
(the reader is referred to \cite{brim2011faster} in order to grasp 
all the details of the proof of correctness and complexity).

\paragraph*{}
\emph{Remark.} 
The Value-Iteration procedure lends itself to the following basic generalization, 
which turns out to be of a pivotal importance in order to best suit our technical needs. 
Let $f^*$ be the least SEPM of the \EG $\Gamma$.
Recall that, as a first step, the Value-Iteration algorithm initializes $f$ to be the constant zero function.
Here, we remark that it is not necessary to do that really. 
Indeed, it is sufficient to initialize $f$ to be any function $f_0$ which bounds $f^*$ from below, 
that is to say, to initialize $f$ to any $f_0:V\rightarrow\C_\Gamma$ such that $f_0(v) \preceq f^*(v)$ for every $v\in V$.
Soon after, $\L$ can be initialized in a natural way: 
just insert $v$ into $\L$ if and only if $f_0$ is inconsistent at $v$.
This initialization still requires $O(|E|)$ time and it doesn't affect the correctness of the procedure.

So, we shall assume to have at our disposal a procedure named \texttt{Value-Iteration()}, 
which takes as input an \EG $\Gamma=(V, E, w, \langle V_0, V_1\rangle )$ 
and an \emph{initial function} $f_0$ that bounds from below the least SEPM $f^*$ of the \EG $\Gamma$ 
(\ie s.t. $f_0(v)\preceq f^*(v)$ for every $v\in V$). 
Then, \texttt{Value-Iteration()} outputs the least SEPM $f^*$ of the \EG $\Gamma$ 
within $O(|V|\,|E|\,W)$ time and working with $O(|V|)$ space.

\section{Values and Optimal Positional Strategies from Reweightings}\label{section:values}
This section aims to show that values and optimal positional strategies of \MPG{s}
admit a suitable description in terms of reweighted arenas.
This fact will be the crux for solving the Value Problem 
and Optimal Strategy Synthesis in $O(|V|^2 |E|\, W)$ time. 

\subsection{On optimal values}
A simple representation of values in terms of \emph{Farey sequences} is now observed, 
then, a characterization of values in terms of reweighted arenas is provided.
\paragraph*{Optimal values and Farey sequences.} 
Recall that each value $\val{\Gamma}{v}$ is contained within the following set of rational numbers:
\[ S_{\Gamma}=\left\{ \frac{N}{D} \suchthat D\in [1, |V|],\, N\in [-DW, DW] \right\}.\]
Let us introduce some notation in order to handle $S_{\Gamma}$ in a way that is suitable for our purposes.
Firstly, we write every $\nu\in S_{\Gamma}$ as $\nu=i+F$, where $i=i_\nu=\lfloor \nu \rfloor$ 
is the integral and $F=F_\nu=\{\nu\}=\nu-i$ is the fractional part.   
Notice that $i\in [-W, W]$ and that $F$ is a non-negative rational number having denominator at most $|V|$.

As a consequence, it is worthwhile to consider the \emph{Farey sequence} $\F_n$ of order $n=|V|$.
This is the increasing sequence of all irreducible 
fractions from the (rational) interval $[0,1]$ with denominators less than or equal to $n$. 
In the rest of this paper, $\F_n$ denotes the following sorted set:
\[ \F_n=\left\{\frac{N}{D} \suchthat 0 \leq N \leq D\leq n, \gcd(N,D)=1 \right\}. \]

Farey sequences have numerous and interesting properties, in particular, many algorithms for 
generating the entire sequence $\F_n$ in time $O(n^2)$ are known in the literature~\cite{GKP:94},  
and these rely on \emph{Stern-Brocot} trees and \emph{mediant} properties.
Notice that the above mentioned quadratic running time is optimal, 
as it is well-known that the sequence $\F_n$ has $s(n) = \frac{3\,n^2}{\pi^2} + O(n\ln n) = \Theta(n^2)$ terms.

Throughout the article, we shall assume that $F_0, \ldots, F_{s-1}$ is an increasing ordering of $\F_n$, 
so that $\F_n=\{F_j\}_{j=0}^{s-1}$ and $F_j < F_{j+1}$ for every $j$.

Also notice that $F_0=0$ and $F_{s-1}=1$. 

For example, $\F_5=\{0, \frac{1}{5}, \frac{1}{4}, \frac{1}{3}, \frac{2}{5}, \frac{1}{2}, \frac{3}{5}, 
\frac{2}{3}, \frac{3}{4}, \frac{4}{5}, 1 \}$. 







At this point, $S_{\Gamma}$ can be represented as follows:
\[S_{\Gamma} = [-W, W) + \F_{|V|} = \left\{i+F_j \suchthat i\in [-W, W),\, j\in [0, s-1]\right\}. \]
The above representation of $S_{\Gamma}$ will be convenient in a while.

\paragraph{Optimal values and reweightings.} 
Two introductory lemmata are shown below, then, a characterization of optimal values in terms of reweightings is provided.
\begin{Lem}\label{Lem:additivity_lemma} 
Let $\Gamma = ( V, E, w, \langle V_0, V_1 \rangle )$ be an \MPG and let $q\in\Q$ 
be a rational number having denominator $D\in\N$. 
Then, $\val{\Gamma}{v} = \frac{1}{D}\val{\Gamma^{w+q}}{v} - q$ holds for every $v\in V$.
\end{Lem}
\begin{proof}
Let us consider the play $\texttt{outcome}^{\Gamma^{w+q}}(v, \sigma_0, \sigma_1)=v_0v_1\ldots v_n\ldots$ 
By the definition of $\val{\Gamma}{v}$, 
and by that of reweighting $\Gamma^{w+q}$ ($= \Gamma^{D\cdot w+N}$), the following holds: 
\[
\begin{array}{lll} 
\val{\Gamma^{w+q}}{v} &=& \sup_{\sigma_0\in\Sigma_0}\inf_{\sigma_1\in\Sigma_1}
\texttt{MP}_0(\texttt{outcome}^{\Gamma^{w+q}}(v, \sigma_0, \sigma_1)) \\ 
 & & \\
 &=& \sup_{\sigma_0\in\Sigma_0}\inf_{\sigma_1\in\Sigma_1}
\lim\inf_{n\rightarrow\infty}\frac{1}{n}\sum_{i=0}^{n-1}(D\cdot w(v_i, v_{i+1}) + N)\;\;  \text{(if $q=N/D$)}\\ 
 & & \\
 &=& D\cdot \sup_{\sigma_0\in\Sigma_0}\inf_{\sigma_1\in\Sigma_1}
\texttt{MP}_0(\texttt{outcome}^{\Gamma}(v, \sigma_0, \sigma_1)) + N\\
 & & \\ 
&=& D\cdot \val{\Gamma}{v} + N. 
\end{array}
\]
Then, $\val{\Gamma}{v} = \frac{1}{D}\val{\Gamma^{w+q}}{v} - \frac{N}{D} = 
 \frac{1}{D}\val{\Gamma^{w+q}}{v} - q$ holds for every $v\in V$.


\end{proof}

\begin{Lem}\label{Lem:transition_lemma}
Given an \MPG $\Gamma = ( V, E, w, \langle V_0, V_1 \rangle )$, let us consider the reweightings:  
\[\Gamma_{i,j}= \Gamma^{w-i-F_j},\, \text{ for any } i\in [-W, W]\text{ and }j\in [0, s-1], \]
where $s=|\F_{|V|}|$ and $F_j$ is the $j$-th term of the Farey sequence $\F_{|V|}$.

Then, the following propositions hold:
\begin{enumerate}
\item\label{lem:transition_lemma:item1} For any $i\in [-W, W]$ and $j\in [0,s-1]$, we have:
\[ v\in \W_0(\Gamma_{i,j}) \text{ if and only if } \val{\Gamma}{v}\geq i+ F_j;\]
\item\label{lem:transition_lemma:item2} For any $i\in [-W, W]$ and $j\in [1,s-1]$, we have: 
\[ v\in \W_1(\Gamma_{i,j}) \text{ if and only if } \val{\Gamma}{v}\leq i + F_{j-1}.\]
\end{enumerate}
\end{Lem}
\begin{proof}\mbox{}

\begin{enumerate}  
\item
Let us fix arbitrarily some $i\in [-W, W]$ and $j\in [0,s-1]$. 

Assume that $F_j = N_j/D_j$ for some $N_j,D_j\in \N$.

Since \[\Gamma_{i,j}=( V, E, D_j(w - i) - N_j, \langle V_0, V_1 \rangle ),\] 
then by Lemma~\ref{Lem:additivity_lemma} (applyed to $q=-i-F_j$) we have: 
\[\val{\Gamma}{v} = \frac{1}{D_j}\val{\Gamma_{i,j}}{v} + i + F_j.\] 
Recall that $v\in \W_0(\Gamma_{i,j})$ if and only if $\val{\Gamma_{i,j}}{v}\geq 0$.  

Hence, we have $v\in \W_0(\Gamma_{i,j})$ if and only if the following inequality holds: 
\begin{align*}
\val{\Gamma}{v} &= \frac{1}{D_j}\val{\Gamma_{i,j}}{v} + i + F_j \\
	        &\geq i + F_j.  
\end{align*}
This proves Item~\ref{lem:transition_lemma:item1}.

\item The argument is symmetric to that of Item~1, but with some further observations. 

Let us fix arbitrarily some $i\in [-W, W]$ and $j\in [1,s-1]$.
Assume that $F_j = N_j/D_j$ for some $N_j,D_j\in \N$.
Since $\Gamma_{i,j}=( V, E, D_j(w - i) - N_j, \langle V_0, V_1 \rangle )$,  
then by Lemma~\ref{Lem:additivity_lemma} we have $\val{\Gamma}{v} = \frac{1}{D_j}\val{\Gamma_{i,j}}{v} + i + F_j$. 
Recall that $v\in \W_1(\Gamma_{i,j})$ if and only if $\val{\Gamma_{i,j}}{v} < 0$.

Hence, we have $v\in \W_1(\Gamma_{i,j})$ if and only if the following inequality holds:
\begin{align*}
\val{\Gamma}{v} & = \frac{1}{D_j}\val{\Gamma_{i,j}}{v} + i + F_j \\ 
		& < i + F_j. 
\end{align*}
Now, recall from Section~\ref{sect:background} that $\val{\Gamma}{v}\in S_{\Gamma}$, 
where \[S_{\Gamma}=\{i+F_j \suchthat i\in [-W, W),\, j\in [0, s-1]\}.\] 
By hypothesis we have: \[j \geq 1 \text{ and } 0\leq F_{j-1} < F_j,\]
thus, at this point, $v\in W_1(\Gamma_{i,j})$ if and only if $\val{\Gamma}{v} \leq i + F_{j-1}$. 

This proves Item~\ref{lem:transition_lemma:item2}.
\end{enumerate}
\end{proof}

We are now in the position to provide a simple characterization of values in terms of reweightings.
\begin{Thm}\label{Thm:transition_opt_values}
Given an \MPG $\Gamma = ( V, E, w, \langle V_0, V_1 \rangle )$, let us consider the reweightings:
\[ \Gamma_{i,j}=\Gamma^{w-i-F_j} ,\,\text{ for any } i\in [-W, W] \text{ and } j\in [1, s-1], \]
where $s=|\F_{|V|}|$ and $F_j$ is the $j$-th term of the Farey sequence $\F_{|V|}$.

Then, the following holds: 
\[  \val{\Gamma}{v} = i + F_{j-1} \text{ if and only if } v\in \W_0(\Gamma_{i,j-1})\cap \W_1(\Gamma_{i, j}).\]
\end{Thm}
\begin{proof}
Let us fix arbitrarily some $i\in [-W, W] $ and $j\in [1, s-1]$.

By Item~1 of Lemma~\ref{Lem:transition_lemma}, we have $v\in \W_0(\Gamma_{i,j-1})$ 
if and only if $\val{\Gamma}{v} \geq i + F_{j-1}$. Symmetrically, by Item~2 of Lemma~\ref{Lem:transition_lemma}, 
we have $v\in \W_1(\Gamma_{i, j})$ if and only if $\val{\Gamma}{v} \leq i + F_{j-1}$.
Whence, by composition, $v\in \W_0(\Gamma_{i,j-1})\cap \W_1(\Gamma_{i, j})$ if and only if $\val{\Gamma}{v} = i + F_{j-1}$.
\end{proof}

\subsection{On optimal positional strategies}
The present subsections aims to provide a suitable description of optimal positional strategies 
in terms of reweighted arenas. An introductory lemma is shown next.
\begin{Lem}\label{Lem:big_lemma} 
Let $\Gamma = ( V, E, w, \langle V_0, V_1 \rangle )$ be an \MPG, 
the following hold:
\begin{enumerate}
\item\label{lem:big_lemma:item1} If $v\in V_0$, let $v'\in \texttt{post}(v)$. Then $\val{\Gamma}{v'}\leq \val{\Gamma}{v}$ holds.
\item\label{lem:big_lemma:item2} If $v\in V_1$, let $v'\in \texttt{post}(v)$. Then $\val{\Gamma}{v'}\geq \val{\Gamma}{v}$ holds.
\item\label{lem:big_lemma:item3} Given any $v\in V_0$, 
consider the reweighted \EG $\Gamma_v= \Gamma^{w-\val{\Gamma}{v}}$. 

Let $f_v:V\rightarrow \C_{\Gamma_v}$ be any SEPM of the \EG $\Gamma_v$ 
such that $v\in V_{f_v}$ (\ie $f_v(v)\neq \top$). Let $v'_{f_v}\in V$ be any vertex out of $v$ 
such that $(v, v'_{f_v})\in E$ is compatible with $f_v$ in $\Gamma_v$. 

Then, $\val{\Gamma}{v'_{f_v}}=\val{\Gamma}{v}$.
\end{enumerate}
\end{Lem}

\begin{proof}
\begin{enumerate}
\item 
It is sufficient to construct a strategy $\sigma^v_0\in\Sigma^M_{0}$ securing to Player~$0$ a
payoff at least $\val{\Gamma}{v'}$ from $v$ in the \MPG $\Gamma$. Let $\sigma^{v'}_0\in\Sigma^M_{0}$ 
be a strategy securing payoff at least $\val{\Gamma}{v'}$ from $v'$ in $\Gamma$.
Then, let $\sigma^v_0$ be defined as follows: 
\[\sigma^v_0(u) = 
\left\{ \begin{array}{ll} \sigma^{v'}_0(u) & ,\text{ if } u\in V_0\setminus\{v\}; \\
 \sigma^{v'}_0(v) & ,\text{ if } u=v \text{ and } v 
\text{ \emph{is} reachable from } v' \text{ in } G^{\Gamma}_{\sigma^{v'}_0}; \\
	 v' & ,\text{ if } u=v \text{ and } v 
\text{ \emph{is not} reachable from } v' \text{ in } G^{\Gamma}_{\sigma^{v'}_0}. 
\end{array}\right.
\]
We argue that $\sigma^v_0$ secures payoff at least $\val{\Gamma}{v'}$ from $v$ in $\Gamma$.
First notice that, by Lemma~\ref{Lem:reachable_cycle} (applied to $v'$), 
all cycles $C$ that are reachable from $v'$ in $\Gamma$ satisfy: \[\frac{w(C)}{|C|}\geq\val{\Gamma}{v'}.\] 
The fact is that any cycle reachable from $v$ in $G^{\Gamma}_{\sigma^v_0}$ is also reachable 
from $v'$ in $G^{\Gamma}_{\sigma^{v'}_0}$ (by definition of $\sigma^v_0$), 
therefore, the same inequality holds for all cycles reachable from $v$. 
At this point, the thesis follows again by 
Lemma~\ref{Lem:reachable_cycle} (applied to $v$, in the inverse direction).
This proves Item~\ref{lem:big_lemma:item1}. 

\item The proof of Item~\ref{lem:big_lemma:item2} is symmetric to that of Item~\ref{lem:big_lemma:item1}.

\item 
Firstly, notice that $\val{\Gamma}{v'_{f_v}}\leq \val{\Gamma}{v}$ holds by Item~\ref{lem:big_lemma:item1}.
To conclude the proof it is sufficient to show $\val{\Gamma}{v'_{f_v}}\geq \val{\Gamma}{v}$.
Recall that $(v, v'_{f_v})\in E$ is compatible with $f_v$ in $\Gamma_v$ by hypothesis, 
that is: \[f_v(v)\succeq f_v(v'_{f_v}) \ominus \big(w(v, v'_{f_v})-\val{\Gamma}{v}\big).\] 
This, together with the fact that $v\in V_{f_v}$ (\ie $f_v(v)\neq\top$) also holds by hypothesis, 
implies that $v'_{f_v}\in V_f$ (\ie $f_v(v'_{f_v}) \neq \top$).
Thus, by Item~1 of Lemma~\ref{Lem:least_energy_prog_measure}, 
$v'_{f_v}$ is a winning starting position of Player~0 in the \EG $\Gamma_v$.
Whence, by Lemma~\ref{Lem:relation_MPG_EG}, it holds that $\val{\Gamma}{v'_{f_v}}\geq \val{\Gamma}{v}$. 
This proves Item~\ref{lem:big_lemma:item3}.
\end{enumerate}
\end{proof}

We are now in position to provide a sufficient condition, for a positional strategy to be optimal, 
which is expressed in terms of reweighted \EG{s} and their SEPM{s}.
\begin{Thm}\label{Thm:pos_opt_strategy}
Let $\Gamma = ( V , E, w, \langle V_0, V_1 \rangle )$ be an \MPG.
For each $v\in V$, consider the reweighted \EG $\Gamma_v = \Gamma^{w-\val{\Gamma}{v}}$. 
Let $f_{v}:V\rightarrow\C_{\Gamma_v}$ be any SEPM of $\Gamma_v$ such that $v\in V_{f_v}$ (\ie $f_v(v)\neq\top$). 
Moreover, assume: $f_{v_1}=f_{v_2}$ whenever $\val{\Gamma}{v_1}=\val{\Gamma}{v_2}$.

When $v\in V_0$, let $v'_{f_v}\in V$ be any vertex such that 
$(v, v'_{f_v})\in E$ is compatible with $f_v$ in the \EG $\Gamma_v$,  
and consider the positional strategy $\sigma^*_0\in\Sigma^M_{0}$ defined as follows: 
\[\sigma^*_0(v) = v'_{f_v} \, \text{ for every } v\in V_0.\]
Then, $\sigma^*_0$ is an optimal positional strategy for Player~$0$ in the \MPG $\Gamma$.
\end{Thm}
\begin{proof} 
Let us consider the projection graph $G_{\sigma^*_0}^{\Gamma}=( V, E_{\sigma^*_0}, w)$. 
Let $v\in V$ be any vertex. In order to prove that $\sigma^*_0$ is optimal,  
it is sufficient (by Lemma~\ref{Lem:reachable_cycle}) to show that 
every cycle $C$ that is reachable from $v$ in $G_{\sigma^*_0}^{\Gamma}$ satisfies $\frac{w(C)}{|C|}\geq \val{\Gamma}{v}$.
\begin{itemize}
\item \emph{Preliminaries.} Let $v \in V$ and let $C$ be any cycle of length $|C|\geq 1$  
that is reachable from $v$ in $G^{\Gamma}_{\sigma^*_0}$.
Then, there exists a path $\rho$ of length $|\rho|\geq 1$ in $G^{\Gamma}_{\sigma^*_0}$ and such that: 
if $|\rho|=1$, then $\rho=\rho_0\rho_1=vv$; otherwise, if $|\rho|>1$, then: 
\[\rho = \rho_0 \ldots \rho_{|\rho|} = v v_1 v_2 \ldots v_k u_1 u_2 \ldots u_{|C|} u_1,\]
where $vv_1\ldots v_k$ is a simple path, for some $k\geq 0$ and $u_1\ldots u_{|C|}u_1=C$.

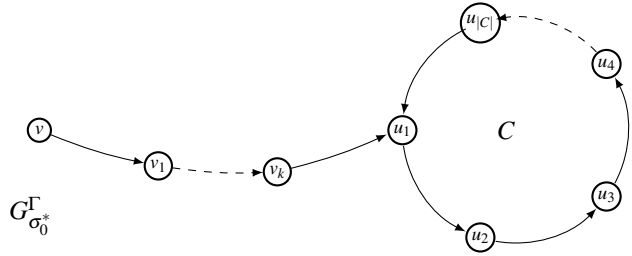
\begin{figure}[!h]
\begin{center}
\begin{tikzpicture}[scale=0.5]

\node[inner sep=0pt,outer sep=0pt,
	label={[xshift=5pt, yshift=-5ex]below right:$G^{\Gamma}_{\sigma^*_0}$}] at (-14,0) (a) {};
\node[inner sep=0pt,outer sep=0pt,
	label={[xshift=5pt, yshift=-5ex]below right:$C$}] at (-1,2) (a') {};
\node[inner sep=0pt,outer sep=0pt] at (-10,2) (b') {};
\node[inner sep=0pt,outer sep=0pt] at (1,4) (b) {};
\node[inner sep=0pt,outer sep=0pt] at (4,4) (c) {};
\node[inner sep=0pt,outer sep=0pt] at (5,2) (d) {};
\node[inner sep=0pt,outer sep=0pt] at (4.5,.5) (e) {};
\node[inner sep=0pt,outer sep=0pt] at (4,-3.3) (f) {};
\node[inner sep=0pt,outer sep=0pt] at (3,-3.3) (g) {};
\node[inner sep=0pt,outer sep=0pt] at (1,-3.3) (h) {};

\def \n {5}
\def \radius {3cm}
\def \margin {8} 

\node[node, thick, xshift=-80ex, font=\Large] (A) {$v$};
\node[node, thick, xshift=-60ex, yshift=-6ex,  font=\Large] (B) {$v_1$};
\node[node, thick, xshift=-40ex, yshift=-7ex, font=\Large] (C) {$v_k$};
\node[xshift=-13ex, xshift=3ex] (D) {};

\draw[->, >=latex] (A) to [bend left=-5] node[above] {} (B); 
\draw[->, >=latex, dashed] (B) to [bend left=-5] node[above] {} (C); 
\draw[->, >=latex] (C) to [bend left=-5] node[above] {} (D); 

\foreach \s in {1,...,3}
{
  \node[node, thick, font=\Large] at ({360/\n * (\s - 1)+180}:\radius) {$u_{\s}$};
  \draw[->, >=latex] ({360/\n * (\s - 1)+180+\margin}:\radius) 
    arc ({360/\n * (\s - 1)+180+\margin}:{360/\n * (\s)+180-\margin}:\radius);
}

\node[node, thick, font=\Large] at ({360/\n * (4 - 1)+180}:\radius) {$u_4$};
\draw[->, >=latex, dashed] ({360/\n * (4 - 1)+180+\margin}:\radius) 
 arc ({360/\n * (4 - 1)+180+\margin}:{360/\n * (4)+180-\margin}:\radius);

\node[node, thick, font=\Large] at ({360/\n * (\n - 1)+180}:\radius) {$u_{|C|}$};
\draw[->, >=latex] ({360/\n * (\n - 1)+180+\margin}:\radius) 
 arc ({360/\n * (\n - 1)+180+\margin}:{360/\n * (\n)+180-\margin}:\radius);

\end{tikzpicture}
\end{center}
\caption{A cycle $C$ that is reachable from $v$ through $v_1\cdots v_k$ in $G^{\Gamma}_{\sigma^*_0}$.}
\end{figure}

\item \emph{Fact~1.} It holds $\val{\Gamma}{\rho_{i}}\leq \val{\Gamma}{\rho_{i+1}}$ for every $i\in [0, |\rho|)$.
\begin{proof}[of Fact~1]
If $\rho_i\in V_0$ then $\val{\Gamma}{\rho_{i}}=\val{\Gamma}{\rho_{i+1}}$ by Item~3 of Lemma~\ref{Lem:big_lemma};
otherwise, if $\rho_i\in V_1$, then $\val{\Gamma}{\rho_{i}}\leq \val{\Gamma}{\rho_{i+1}}$ 
by Item~2 of Lemma~\ref{Lem:big_lemma}. This proves Fact~1. In particular, 
notice that $\val{\Gamma}{v}\leq \val{\Gamma}{u_1}$ when $|\rho| > 1$. 
\end{proof}

\item \emph{Fact~2.} Assume $C=u_1\ldots u_{|C|}u_1$, then $\val{\Gamma}{u_i}=\val{\Gamma}{u_1}$ for every $i\in [0, |C|] $.
\begin{proof}[of Fact~2]
By Fact~1, $\val{\Gamma}{u_{i-1}}\leq\val{\Gamma}{u_i}$ for every $i\in [2, |C|]$, 
as well as $\val{\Gamma}{u_{|C|}}\leq \val{\Gamma}{u_1}$. Then, the following chain of inequalities holds: 
\[ \val{\Gamma}{u_1}\leq\val{\Gamma}{u_2} \leq \ldots \leq \val{\Gamma}{u_{|C|}} \leq \val{\Gamma}{u_1}. \]
Since the first and the last value of the chain are actually the same, \ie $\val{\Gamma}{u_1}$,  
then, all these inequalities are indeed equalities. This proves Fact~2.
\end{proof}

\item \emph{Fact~3.} The following holds for every $i\in [0, |\rho|)$: 
\[ f_{\rho_i}(\rho_i), f_{\rho_i}(\rho_{i+1})\neq\top \text{ and }  
f_{\rho_i}(\rho_i) \geq f_{\rho_i}(\rho_{i+1}) - w(\rho_i, \rho_{i+1}) + \val{\Gamma}{\rho_{i}}.\]
\begin{proof}[of Fact~3]
Firstly, we argue that any arc $(\rho_i, \rho_{i+1})\in E$ is compatible with $f_{\rho_i}$ in $\Gamma_{\rho_i}$. 
Indeed, if $\rho_i\in V_0$, then $(\rho_i, \rho_{i+1})$ is compatible with $f_{\rho_i}$ in $\Gamma_{\rho_i}$ 
because $\rho_{i+1}=\sigma^*_0(\rho_i)$ by hypothesis;
otherwise, if $\rho_i\in V_1$, 
then $( \rho_i, x)$ is compatible with $f_{\rho_i}$ in $\Gamma_{\rho_i}$ 
for \emph{every} $x\in \texttt{post}(\rho_{i})$, in particular for $x=\rho_{i+1}$, by definition of SEPM. 

At this point, since $(\rho_i, \rho_{i+1})$ is compatible with $f_{\rho_i}$ in $\Gamma_{\rho_i}$, then: 
\[f_{\rho_i}(\rho_i) \succeq f_{\rho_i}(\rho_{i+1}) \ominus \big(w(\rho_i, \rho_{i+1}) - \val{\Gamma}{\rho_{i}}\big).\] 
Now, recall that $\rho_i\in V_{f_{\rho_i}}$ (\ie $f_{\rho_i}(\rho_i)\neq\top$) holds for every $\rho_i$ by hypothesis.
Since $f_{\rho_i}(\rho_i)\neq\top$ and the above inequality holds, 
then we have $f_{\rho_i}(\rho_{i+1})\neq\top$. Thus, we can safely write:
\[f_{\rho_i}(\rho_i) \geq f_{\rho_i}(\rho_{i+1}) - w(\rho_i, \rho_{i+1}) + \val{\Gamma}{\rho_{i}}.\]
This proves Fact~3.
\end{proof}

\item \emph{Fact~4.} 
Assume that the cycle $C=u_1\ldots u_{|C|}u_1$ is such that: 
\[\val{\Gamma}{u_i}=\val{\Gamma}{u_1}\geq \val{\Gamma}{v}, \text{ for every } i\in [1,|C|].\]

Then, provided that $u_{|C|+1}= u_1$, the following holds for every $i\in [1, |C|]$:
\[ f_{u_1}(u_1), f_{u_{i+1}}(u_{i+1})\neq\top \text{ and } 
f_{u_1}(u_1) \geq f_{u_{i+1}}(u_{i+1}) - \sum_{j=1}^{i} w(u_j, u_{j+1}) + i\cdot\val{\Gamma}{v}.\] 

\begin{proof}[of Fact~4]
Firstly, notice that $f_{u_1}(u_1), f_{u_{i+1}}(u_{i+1})\neq\top$ holds by hypothesis.  

The proof proceeds by induction on $i\in [1, |C|]$. 
\begin{itemize}
\item \emph{Base Case.} 
Assume that $|C|=1$, so that $C=u_1u_1$. 
Then $f_{u_1}(u_1)\geq f_{u_1}(u_1) - w(u_1, u_1) + \val{\Gamma}{u_1}$ follows by Fact~3.
Since $\val{\Gamma}{u_1}\geq \val{\Gamma}{v}$ by hypothesis, then the thesis follows.

\item \emph{Inductive Step.} Assume by induction hypothesis that the following holds: 
\[f_{u_1}(u_1) \geq f_{u_{i}}(u_{i}) - \sum_{j=1}^{i-1} w(u_j, u_{j+1}) + (i-1)\cdot \val{\Gamma}{v}.\] 
By Fact~3, we have: \[f_{u_{i}}(u_{i})\geq f_{u_{i}}(u_{i+1}) - w(u_{i}, u_{i+1}) + \val{\Gamma}{u_{i}}.\]
Since $\val{\Gamma}{u_{i+1}}=\val{\Gamma}{u_{i}}$ holds by hypothesis, then we have $f_{u_{i+1}} = f_{u_i}$. 
Recall that $\val{\Gamma}{u_i}\geq \val{\Gamma}{v}$ also holds by hypothesis.

Thus, we obtain the following: 
\[f_{u_1}(u_1) \geq f_{u_{i+1}}(u_{i+1}) - \sum_{j=1}^{i} w(u_j, u_{j+1}) + i\cdot\val{\Gamma}{v}.\] 
This proves Fact~4.
\end{itemize}
\end{proof}

\item We are now in position to show that every cycle $C$ that is reachable from 
$v$ in $G^{\Gamma}_{\sigma^*_0}$ satisfies $w(C)/|C|\geq \val{\Gamma}{v}$. By Fact~1 and Fact~2, we have 
$\val{\Gamma}{v}\leq \val{\Gamma}{u_1}=\val{\Gamma}{u_i}$ for every $i\in [1, |C|]$. At this point, we apply Fact~4.
Consider the specialization of Fact~4 when $i=|C|$ and also recall that $u_{|C|+1}=u_1$. Then, we have the following:
\[f_{u_1}(u_1) \geq f_{u_1}(u_1) - \sum_{j=1}^{|C|} w(u_j, u_{j+1}) + |C|\cdot \val{\Gamma}{v}.\]
As a consequence, the following lower bound holds on the average weight of $C$: 
\[\frac{w(C)}{|C|} = \frac{1}{|C|}\sum_{j=1}^{|C|} w(u_j, u_{j+1}) \geq \val{\Gamma}{v}, \]
which concludes the proof.
\end{itemize}
\end{proof}

\begin{Rem}\label{Rem:pos_opt_strategy} 
Notice that Theorem~\ref{Thm:pos_opt_strategy} holds, in particular, 
when $f_v$ is the least SEPM $f^*_v$ of the reweighted \EG $\Gamma_v$.
This follows because $v\in V_{f^*_v}$ always holds for the least SEPM $f^*_v$ of the \EG $\Gamma_v$, as shown next:
by Lemma~\ref{Lem:relation_MPG_EG} and by definition of $\Gamma_v$, 
then $v$ is a winning starting position for Player~0 in the \EG $\Gamma_v$ (for some initial credit);
now, since $f^*_v$ is the least SEPM of the \EG $\Gamma_v$, then $v\in V_{f^*_v}$ 
follows by Item~2 of Lemma~\ref{Lem:least_energy_prog_measure}.
\end{Rem}

\section{An $O(|V|^2 |E|\, W)$ time Algorithm for solving 
the Value Problem and Optimal Strategy Synthesis in \MPG{s}}\label{sect:algorithm}
This section offers a deterministic algorithm for solving the Value Problem and Optimal Strategy Synthesis 
of \MPG{s} within $O(|V|^2 |E|\, W)$ time and $O(|V|)$ space, 
on any input \MPG $\Gamma= ( V, E, w, \langle V_0, V_1\rangle )$.

Let us now recall some notation in order describe the algorithm in a suitable way.

Given an \MPG $\Gamma = ( V, E, w, \langle V_0, V_1 \rangle )$, 
consider again the following reweightings:
\[ \Gamma_{i,j}=\Gamma^{w-i-F_j},\,\text{ for any } i\in [-W, W] \text{ and } j\in [0, s-1], \]
where $s=|\F_{|V|}|$ and $F_j$ is the $j$-th term of $\F_{|V|}$.

Assuming $F_j=N_j/D_j$ for some $N_j,D_j\in \N$, we focus on the following  
weights: 
\begin{align*}
w_{i,j}= & w-i-F_j = w-i-\frac{N_j}{D_j}; \\ 
w'_{i,j}= & D_j\, w_{i,j} = D_j\, (w-i) - N_j. 
\end{align*}
Recall that $\Gamma_{i,j}$ is defined as $\Gamma_{i,j}:=\Gamma^{w'_{i,j}}$, which is an arena having integer weights.
Also notice that, since $F_0 < \ldots < F_{s-1}$ is monotone increasing, 
then the corresponding weight functions $w_{i,j}$ can be ordered in a natural way, \ie 
$w_{-W, 1} > w_{-W, 2} > \ldots > w_{W-1, s-1} >\ldots > w_{W, s-1}$. 
In the rest of this section, 
we denote by $f^*_{w'_{i,j}}:V\rightarrow \C_{\Gamma_{i,j}}$ the least SEPM of the reweighted \EG $\Gamma_{i,j}$. 
Moreover, the function $f^*_{i,j}:V\rightarrow\Q$, 
defined as $f^*_{i,j}(v):=\frac{1}{D_j}f^*_{w'_{i,j}}(v)$ for every $v\in V$, 
is called the \emph{rational scaling} of $f^*_{w'_{i,j}}$.

\subsection{Description of the Algorithm} 
In this section we shall describe a procedure whose pseudo-code is given below in Algorithm~\ref{Algo:solve_mpg}. 
It takes as input an arena $\Gamma=(V, E, w, \langle V_0, V_1 \rangle)$,  
and it aims to return a tuple $( \W_0, \W_1, \nu, \sigma^*_0 )$ such that: 
$\W_0$ and $\W_1$ are the winning regions of Player~$0$ and Player~$1$ in the \MPG $\Gamma$ (respectively), 
$\nu:V\rightarrow S_\Gamma$ is a map sending each starting position $v\in V$ to its optimal value, \ie $\nu(v)=\val{\Gamma}{v}$, 
and finally, $\sigma^*_0:V_0\rightarrow V$ is an optimal positional strategy for Player $0$ in the \MPG $\Gamma$. 

The intuition underlying Algorithm~\ref{Algo:solve_mpg} is that of considering the following sequence of weights:
\[w_{-W, 1} > w_{-W, 2} > \ldots > w_{-W, s-1} > w_{-W+1, 1} > w_{-W+1, 2} > \ldots  > w_{W-1, s-1} > \ldots > w_{W, s-1} \]
where the key idea is that to rely on Theorem~\ref{Thm:transition_opt_values} at each one of these steps,
testing whether a \emph{transition of winning regions} has occurred. 
\begin{figure}[!h]
\begin{center}
\begin{tikzpicture}[xscale=1.0]
\draw[-][draw=blue, very thick] (0,0) -- (4,0);
\draw[-][draw=green, very thick] (4,0) -- (5,0);
\draw[-][draw=blue, very thick]  (5,0) -- (10,0);

\node[node, color=red, fill=red, scale=.3] (A) at (0,.1) {};
\node[node, left=of A, xshift=3ex, color=black, fill=black, scale=.2, label={above:start} ] (S) {};
\node[node, color=red, fill=red, scale=.3] (B) at (1,.1) {};
\node[node, color=red, fill=red, scale=.3] (C) at (2,.1) {};
\node[node, color=red, fill=red, scale=.3] (D) at (3,.1) {};
\node[node, color=green, fill=green, scale=.3] (E) at (4,.1) {};
\node[node, color=green, fill=green, scale=.3] (F) at (5,.1) {};
\node[-] (G) at (6,.1) {};

\draw[thick, arrows=->] (S) to [ bend left=0] node[below, xshift=-1.5ex, yshift=1ex] {} (A);
\draw[thick, arrows=->] (A) to [ bend left=50] node[below, xshift=-1.5ex, yshift=1ex] {} (B); 
\draw[thick, arrows=->] (B) to [ bend left=50] node[below, xshift=-1.5ex, yshift=1ex] {} (C); 
\draw[thick, arrows=->] (C) to [ bend left=50] node[below, xshift=-1.5ex, yshift=1ex] {} (D); 
\draw[thick, arrows=->] (D) to [ bend left=50] node[below, xshift=-1.5ex, yshift=1ex] {} (E); 
\draw[thick, dashed, arrows=->] (E) to [bend left=50] node[below, xshift=-1.5ex, yshift=1ex, 
label={above right:$v\overset{?}{\in} \W_0(\Gamma_{\texttt{prev}(i,j)}) \cap \W_1(\Gamma_{i,j})$}] {} (F); 

\draw [thick] (0,-.1) node[below]{$w_{-W, 1}$} -- (0,0.1);
\draw [thick] (1,-.1) node[below]{$w_{-W, 2}$} -- (1,0.1);
\draw [thick] (2,-.1) node[below]{$w_{-W, 3}$} -- (2,0.1);
\draw [thick] (3,-.1) node[below]{$\cdots$} -- (3,0.1);
\draw [thick] (4,-.1) node[below]{$w_{\texttt{prev}(i,j)}$} -- (4,0.1);
\draw [thick] (5,-.1) node[below]{$w_{i, j}$} -- (5,0.1);
\draw [thick] (7,-.1) node[below, xshift=2ex]{$\cdots$} -- (7,-.1);
\draw [thick] (7,-.1) node[above, xshift=-2ex, yshift=1ex]{$\cdots$} -- (7,-.1);
\draw [thick] (9,-.1) node[below]{$w_{W-1, s-1}$} -- (9,0.1);
\draw [thick] (10,-.1) node[below]{$w_{W, 1}$} -- (10,0.1);
\end{tikzpicture}
\end{center}
\caption{An illustration of Algorithm~\ref{Algo:solve_mpg}.}\label{fig:algo_solve}
\end{figure}
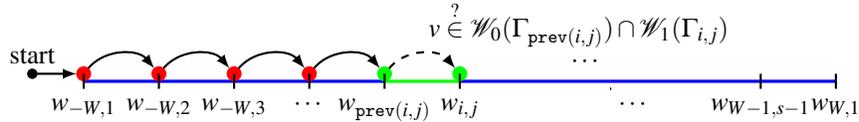
Stated otherwise, the idea is to check, for each vertex $v\in V$, whether $v$ is winning for Player~$1$  
with respect to the current weight $w_{i,j}$, meanwhile recalling whether $v$ was winning for Player~$0$ with respect 
to the immediately preceding element $w_{\texttt{prev}(i,j)}$ in the weight sequence above. 

If such a transition occurs, say for some $\hat{v} \in \W_0(\Gamma_{\texttt{prev}(i,j)}) \cap \W_1(\Gamma_{i,j})$, 
then one can easily compute $\val{\Gamma}{\hat{v}}$ by relying on Theorem~\ref{Thm:transition_opt_values}; 
Also, at that point, it is easy to compute an optimal positional strategy, 
provided that $\hat{v}\in V_0$, by relying on Theorem~\ref{Thm:pos_opt_strategy} and 
Remark~\ref{Rem:pos_opt_strategy} in that case.

Each one of these phases, in which one looks at transitions of winning regions, is named \emph{Scan Phase}.
A graphical intuition of Algorithm~\ref{Algo:solve_mpg} is given in \figref{fig:algo_solve}.


An in-depth description of the algorithm and of its pseudo-code now follows.

\begin{algorithm}[H]\label{Algo:solve_mpg}
\caption{Solving the Value Problem and Strategy Synthesis in \MPG{s}.}
\DontPrintSemicolon
\nonl \SetKwProg{Fn}{Procedure}{}{}
\Fn{$\texttt{solve\_MPG}(\Gamma)$}{
    \SetKwInOut{Input}{input}
    \SetKwInOut{Output}{output}

\Input{an \MPG $\Gamma= ( V, E, w, \langle V_0, V_1 \rangle )$.}
\Output{a tuple $( \W_0, \W_1, \nu, \sigma^*_0)$ such that: 
$\W_0$ and $\W_1$ are the winning regions of Player~$0$ and Player~$1$ (respectively) in the \MPG $\Gamma$; 
$\nu:V\rightarrow S_\Gamma$ is a map sending each starting position $v\in V$ 
to its corresponding optimal value, \ie $\nu(v)=\val{\Gamma}{v}$;
and $\sigma^*_0:V_0\rightarrow V$ is an optimal positional strategy for Player $0$ in the \MPG $\Gamma$.
}
\tcp{Init Phase}
$\W_0\leftarrow\emptyset$; $\W_1\leftarrow\emptyset$; \label{algo:solve:l1}\; 
$f(v)\leftarrow 0$, $\forall\; v\in V$; \label{algo:solve:l2} \; 
$W\leftarrow \max_{e\in E} |w_e|$; $w'\leftarrow w+W$; $D\leftarrow 1$; \label{algo:solve:l3}\;
$s\leftarrow$ compute the size $|\F_{|V|}|$ of $\F_{|V|}$; \label{algo:solve:l4} 
\tcp{with the algorithm of~\cite{PaPa09}}
\tcp{Scan Phases}
\For{$i=-W$ \textbf{to} $W$}{ \label{algo:solve:l6}
	$F\leftarrow 0$;\; 
	\For{$j=1$ \textbf{to} $s-1$}{ \label{algo:solve:l7}
		$\text{prev\_}f\leftarrow f$; \label{algo:solve:l8}\;	
		$\text{prev\_}w\leftarrow \frac{1}{D}\, w'$; \label{algo:solve:l9} \;	
		$\text{prev\_}F\leftarrow F$;\;
		$F\leftarrow$ generate the $j$-th 
			term of $\F_{|V|}$; \label{algo:solve:l5} \tcp{with the algorithm of~\cite{PaPa09}}
		$N\leftarrow$ numerator of $F$; \label{algo:solve:l10}\;
		$D\leftarrow$ denominator of $F$; \label{algo:solve:l11}\;
		$w'\leftarrow D\, (w - i)-N$; \label{algo:solve:l12}\;
		$f\leftarrow \frac{1}{D}\,\texttt{Value-Iteration}(\Gamma^{w'}, 
			\lceil D\, \text{prev\_}f \rceil)$; \label{algo:solve:l13}\;
		\For{$v\in V$}{ \label{algo:solve:l14}
			\If{$\text{prev\_}f(v)\neq\top$ \textbf{and} $f(v)=\top$}{ \label{algo:solve:l15}
				$\nu(v)\leftarrow i + \texttt{prev\_}F$; \label{algo:solve:l16} \tcp{set optimal value $\nu$}
				\If{$\nu(v)\geq 0$}{ \label{algo:solve:l17}
					$\W_0\leftarrow \W_0\cup \{v\}$; \label{algo:solve:l18} \tcp{$v$ is winning for Player~$0$}
				}\Else{ \label{algo:solve:l19}
					$\W_1\leftarrow \W_1\cup \{v\}$; \label{algo:solve:l20} \tcp{$v$ is winning for Player~$1$}
				}
				\If{$v\in V_0$}{ \label{algo:solve:l21}
					\For{$u\in \texttt{post}(v)$}{ \label{algo:solve:l22}
				\If{$\text{prev\_}f(v)\succeq \text{prev\_}f(u) \ominus \text{prev\_}w(v, u)$}{ \label{algo:solve:l23}
	$\sigma^*_0(v)\leftarrow u$; \textbf{break}; \label{algo:solve:l24}
						}
					}
				}
			}
		}
	}
}
\Return{$( \W_0, \W_1, \nu, \sigma^*_0)$} \label{algo:solve:l25}
}
\end{algorithm}
\begin{itemize}
\item\textbf{Initialization Phase.} To start with, the algorithm performs an initialization phase.
At line~\ref{algo:solve:l1}, Algorithm~\ref{Algo:solve_mpg} initializes the output variables $\W_0$ and $\W_1$ to be empty sets.
Notice that, within the pseudo-code, the variables $\W_0$ and $\W_1$ 
represent the winning regions of Player~$0$ and Player~$1$, respectively; 
also, the variable $\nu$ represents the optimal values of the input \MPG $\Gamma$,  
and $\sigma^*_0$ represents an optimal positional strategy for Player~0 in the input \MPG $\Gamma$. 
Secondly, at line~\ref{algo:solve:l2}, 
an array variable $f:V\rightarrow\C_{\Gamma}$ is initialized to $f(v)=0$ for every $v\in V$;
throughout the computation, the variable $f$ represents a SEPM. 
Next, at line~\ref{algo:solve:l3}, the greatest absolute weight $W$ is assigned as $W=\max_{e\in E}|w_e|$, 
an auxiliary weight function $w'$ is initialized as $w'=w+W$, and a ``denominator" variable is initialized as $D=1$.
Concluding the initialization phase, 
at line~\ref{algo:solve:l4} the size (\ie the total number of terms) of $\F_{|V|}$ is computed and assigned to the variable $s$.
This size can be computed very efficiently with the algorithm devised by Pawlewicz and P\u{a}tra\c{s}cu~\cite{PaPa09}.  

\item\textbf{Scan Phases.} After initialization, the procedure performs multiple \emph{Scan Phases}.
Each one of these is \emph{indexed} by a pair of integers $(i,j)$, 
where $i\in [-W, W]$ (at line~\ref{algo:solve:l6}) and $j\in [1, s-1]$ (at line~\ref{algo:solve:l7}). 
Thus, the index $i$ goes from $-W$ to $W$,  
and for each $i$, the index $j$ goes from $1$ to $s-1$.

At each step, we say that the algorithm goes through the $(i,j)$-th scan phase. 
For each scan phase, we also need to consider the \emph{previous} scan phase, 
so that the \emph{previous index} $\texttt{prev}(i,j)$ 
shall be defined as follows: the predecessor of the first index is $\texttt{prev}(-W,1) := (-W, 0)$;
if $j>1$, then $\texttt{prev}(i,j) := (i,j-1)$; 
finally, if $j=1$ and $i>-W$, then $\texttt{prev}(i,j) := (i-1, s-1)$.

At the $(i,j)$-th scan phase, the algorithm considers the rational 
number $z_{i,j}\in S_\Gamma$ defined as: \[z_{i,j} := i+F[j],\] 
where $F[j]=N_j/D_j$ is the $j$-th term of $\F_{|V|}$. 
For each $j$, $F[j]$ can be computed very efficiently, on the fly,  
with the algorithm of Pawlewicz and P\u{a}tra\c{s}cu~\cite{PaPa09}.
Notice that, 
since $F[0] < \ldots < F[s-1]$ is monotonically increasing, 
then the values $z_{i,j}$ are scanned in increasing order as well.
At this point, the procedure aims to compute the rational scaling $f^*_{i,j}$ 
of the least SEPM $f^*_{w'_{i,j}}$, \ie \[f := f^*_{i,j} =\frac{1}{D_j} f^*_{w'_{i,j}}.\]
This computation is really at the heart of the algorithm and it goes  
from line~\ref{algo:solve:l8} to line~\ref{algo:solve:l13}.
To start with, at line~\ref{algo:solve:l8} and line~\ref{algo:solve:l9}, 
the previous rational scaling $f^*_{\texttt{prev}(i,j)}$ 
and the previous weight function $w_{\texttt{prev}(i,j)}$ 
(\ie those considered during the previous scan phase) 
are saved into the auxiliary variables $\texttt{prev\_}f$ and $\texttt{prev\_}w$. 

\emph{Remark.} Since the values $z_{i,j}$ are scanned in increasing order of 
magnitude, then $\texttt{prev\_}f=f^*_{\texttt{prev}(i,j)}$ bounds from below $f^*_{i,j}$. 
That is, it holds for every $v\in V$ that: 
\[\texttt{prev\_}f(v)=f^*_{\texttt{prev}(i,j)}(v)\preceq f^*_{i,j}.\] 
The underlying intuition, at this point, is that of computing the energy levels of $f=f^*_{i,j}$ 
firstly by initializing them to the energy levels of the previous scan phase, 
\ie to $\texttt{prev\_}f=f^*_{\texttt{prev}(i,j)}$, 
and then to update them monotonically 
upwards by executing the Value-Iteration algorithm for \EG{s}.

Further details of this pivotal step now follow.
Firstly, since the Value-Iteration has been designed  
to work with integer numerical weights only~\cite{brim2011faster}, 
then the weights $w_{i,j}=w-z_{i,j}$ have to be scaled from $\Q$ to $\Z$:
this is performed in the standard way, from line~\ref{algo:solve:l10} to line~\ref{algo:solve:l13},  
by considering the numerator $N_j$ and the denominator $D_j$ of $F[j]$, 
and then by setting: 
\[w'_{i,j}(e):=D_j\, \big(w(e)-i\big)-N_j \text{, for every } e\in E.\] 
The initial energy levels are also scaled up from   
$\Q$ to $\Z$ by considering the values: 
$\lceil D_j\, \texttt{prev\_}f(v)\rceil$, for every $v\in V$ (line~\ref{algo:solve:l13}).
At this point the least SEPM of $\Gamma^{w'_{i,j}}$ is computed, at line~\ref{algo:solve:l13}, 
by invoking $\texttt{Value-Iteration}(\Gamma^{w'_{i,j}}, \lceil D_j\, \texttt{prev\_}f\rceil)$,  
that is, by executing on input $\Gamma^{w'_{i,j}}$ the Value-Iteration  
with initial energy levels given by: $\lceil D_j\,\texttt{prev\_}f(v)\rceil$ for every $v\in V$. 
Soon after that, the energy levels have to be scaled back from $\Z$ to $\Q$, 
so that, in summary, at line~\ref{algo:solve:l13} they becomes: 
\[f=f^*_{i,j}=\frac{1}{D_j}\,\texttt{Value-Iteration}(\Gamma^{w'_{i,j}}, \lceil D_j\, \texttt{prev\_}f\rceil).\] 
The correctness of lines~\ref{algo:solve:l12}-\ref{algo:solve:l13} 
will be proved in Lemma~\ref{lemma:correctness}. 

Here, let us provide a sketch of the argument:
\begin{enumerate}
\item Since $F_0 < \ldots < F_{s-1}$ is monotone increasing, 
then the sequence $\{w'_{i,j}\}_{(i,j)}$ is monotone decreasing, \ie for every $i,j$ and $e\in E$, 
$w'_{\texttt{prev}(i,j)}(e) > w'_{i,j}(e)$.
Whence, the sequence of rational scalings $\{f^*_{i,j}\}_{i,j}$ is monotone increasing, 
\ie $f^*_{i, j}\preceq f^*_{\texttt{prev}(i, j)}$ holds at the $(i,j)$-th step. 
The proof is in Lemma~\ref{lemma:correctness}.

\item At the $(i,j)$-th iteration of line~\ref{algo:solve:l8}, 
it holds that $\texttt{prev\_}f=f^*_{\texttt{prev}(i, j)}$.

This invariant property is also proved as part of Lemma~\ref{lemma:correctness}.
\item Since $\texttt{prev\_}f=f^*_{\texttt{prev}(i,j)}$, 
then $\texttt{prev\_}f\preceq f^*_{i,j}$.

Thus, one can prove that $D_j\, \texttt{prev\_f}\preceq f^*_{w'_{i,j}}$.
\item Since $w'_{i,j}(e)\in \Z$ for every $e\in E$, then $f^*_{w'_{i,j}}(v)\in \Z$ for every $v\in V$, 
so that $\lceil D_j\, \texttt{prev\_f}(v)\rceil \preceq f^*_{w'_{i,j}}(v)$ holds for every $v\in V$ as well. 
\item This implies that it is correct to execute the Value-Iteration, on input $\Gamma^{w'_{i,j}}$, 
with initial energy levels given by: $\lceil D_j\,\texttt{prev\_}f(v)\rceil$ for every $v\in V$.
\end{enumerate}

Back to us, once $f=f^*_{i,j}$ has been determined, then for each $v\in V$ the condition: 
\[v \overset{?}{\in} \W_0(\Gamma_{\texttt{prev}(i,j)})\cap \W_1(\Gamma_{i,j}),\] is checked at line~\ref{algo:solve:l15}: 
it is not difficult to show that, for this, 
it is sufficient to test whether both $\texttt{prev\_f}(v)\neq\top$ and $f(v)=\top$ hold on $v$ (it follows by Lemma~\ref{lemma:correctness}).

If $v\in \W_0(\Gamma_{\texttt{prev}(i,j)})\cap \W_1(\Gamma_{i,j})$ holds, 
then the algorithm relies on Theorem~\ref{Thm:transition_opt_values} 
in order to assign the optimal value as follows: $\nu(v):=\val{\Gamma}{v}=z_{\texttt{prev}(i,j)}$ (line~\ref{algo:solve:l16}).
If $\nu(v)\geq 0$, then $v$ is added to the winning region $\W_0$ at line~\ref{algo:solve:l18}. 
Otherwise, $\nu(v)<0$ and $v$ is added to $\W_1$ at line~\ref{algo:solve:l20}.

To conclude, from line~\ref{algo:solve:l21} to line~\ref{algo:solve:l25}, 
the algorithm proceeds as follows: if $v\in V_0$, 
then it computes an optimal positional strategy $\sigma^*_0(v)$ for Player~$0$ in $\Gamma$: 
this is done by testing for each $u\in \texttt{post}(v)$ whether $(v,u)\in E$ is an arc  
compatible with $\texttt{prev\_}f$ in $\Gamma_{\texttt{prev}(i,j)}$; 
namely, whether the following holds for some $u\in \texttt{post}(v)$: 
\[\text{prev\_} f(v) \overset{?}{\succeq} \text{prev\_}f(u) \ominus \text{prev\_}w(v, u).\] 
If $(v,u)\in E$ is found to be compatible with $\texttt{prev\_}f$ at that point, 
then $\sigma^*_0(v) := u$ gets assigned and the arc $(v,u)$ 
becomes part of the optimal positional strategy returned to output. 
Indeed, the correctness of such an assignment 
relies on Theorem~\ref{Thm:pos_opt_strategy} and Remark~\ref{Rem:pos_opt_strategy}. 

This concludes the description of the scan phases and also that of Algorithm~\ref{Algo:solve_mpg}.
\end{itemize}

\subsection{Proof of Correctness}
Now we formally prove the correctness of Algorithm~\ref{Algo:solve_mpg}.
The following lemma shows some basic invariants that are maintained throughout the computation.
\begin{Lem}\label{lemma:correctness}
Algorithm~\ref{Algo:solve_mpg} keeps the following invariants throughout the computation: 
\begin{enumerate} 
\item\label{lemma:correctness:item1}For every $i\in [-W, W]$ and every $j\in [1, s-1]$, it holds that: 
\[f^*_{\texttt{prev}(i,j)}(v) \preceq f^*_{i,j}(v), \text{ for every } v\in V;\]
\item\label{lemma:correctness:item:2} At the $(i,j)$-th iteration of line~\ref{algo:solve:l8}, it holds that: 
$\texttt{prev\_}f=f^*_{\texttt{prev}(i, j)}$; 
\item\label{lemma:correctness:item:3} At the $(i,j)$-th iteration of line~\ref{algo:solve:l8}, it holds that: 
$\lceil D_j \texttt{prev\_}f \rceil\preceq f^*_{{w'_{i,j}}}$; 
\item\label{lemma:correctness:item:4} At the $(i,j)$-th iteration of line~\ref{algo:solve:l13}, it holds that: 
\[\frac{1}{D_j} \texttt{Value-Iteration}(\Gamma^{w'_{i,j}}, \lceil D_j \texttt{prev\_}f \rceil) = f^*_{i,j}.\] 
\end{enumerate}
\end{Lem}
\begin{proof}\mbox{}

\begin{itemize}
\item \emph{Proof (of Item~\ref{lemma:correctness:item1}).} 
Recall that $w_{i,j}:=w-i-F_j$. Since $F_0< \ldots < F_{s-1}$ is monotone increasing, then: 
$w_{i,j}(e) < w_{\texttt{prev}(i,j)}(e)$ holds for every $e \in E$.

In order to prove the thesis, consider the following function: 
\[ g:V\rightarrow \Q\cup \{\top\} : v\mapsto \min\big( f^*_{\texttt{prev}(i,j)}(v), f^*_{i,j}(v)\big).\] 
We show that $D_{\texttt{prev}(i,j)}\, g$ is a SEPM of $\Gamma^{w'_{\texttt{prev}(i,j)}}$. 
There are four cases, according to whether $v\in V_0$ or $v\in V_1$, 
and $g(v)= f^*_{\texttt{prev}(i,j)}(v)$ or $g(v)= f^*_{i,j}(v)$.
\begin{itemize}
\item \emph{Case: $v\in V_0$.} Then, the following holds for \emph{some} $u\in\texttt{post}(v)$:
\begin{itemize}
\item \emph{Case: $g(v)= f^*_{\texttt{prev}(i,j)}(v)$:}  
\begin{align*}
	D_{\texttt{prev}(i,j)} \, g(v) &= D_{\texttt{prev}(i,j)} \,  f^*_{\texttt{prev}(i,j)}(v)  
	\;\;\;\;\;\;\;\;\;\;\;\;\;\;\;\;\;\;\;\;\;\;\;\;\;\;\;\;\;\;
		\text{[by $g(v) = f^*_{\texttt{prev}(i,j)}(v)$]}\\
 		& = f^*_{w'_{\texttt{prev}(i,j)}}(v) 
	\;\;\;\;\;\;\;\;\;\;\;\;\;\;\;\;\;\;\;\;\;\;\;\;\;\;\;\;\;\;\;\;\;\;\;\;\;\;\;\;\;\;\;\;\;\;\;
		\text{[by $ D_{\texttt{prev}(i,j)}f^*_{\texttt{prev}(i,j)} = f^*_{w'_{\texttt{prev}(i,j)}}$]} \\
		& \succeq f^*_{w'_{\texttt{prev}(i,j)}}(v) \ominus w'_{\texttt{prev}(i,j)}(v,u)  
		\;\;\;\;\;\;\;\;\;\;\;\;\;\;\;\;
		\text{ [$f^*_{w'_{\texttt{prev}(i,j)}}$ is SEPM of $\Gamma^{w'_{\texttt{prev}(i,j)}}$] } \\ 
		& = D_{\texttt{prev}(i,j)} \, f^*_{\texttt{prev}(i,j)}(u) 
		\ominus w'_{\texttt{prev}(i,j)}(v,u) \; 
		\text{[by $ f^*_{w'_{\texttt{prev}(i,j)}} = D_{\texttt{prev}(i,j)}f^*_{\texttt{prev}(i,j)}$]} \\ 
		& \succeq D_{\texttt{prev}(i,j)} \, g(u) \ominus w'_{\texttt{prev}(i,j)}(v,u) 
		\;\;\;\;\;\;\;\;\;\;\;\;\;\; \text{[by definition of $g(u)$]} 
\end{align*}
\item \emph{Case: $g(v) = f^*_{i,j}(v)$:}
\begin{align*}
	D_{\texttt{prev}(i,j)} \, g(v) & = D_{\texttt{prev}(i,j)} \, f^*_{i,j}(v)
			\;\;\;\;\;\;\;\;\;\;\;\;\;\;\;\;\;\;\;\;\;\;\;\;\;\;\;\;\;\;\;\;\;\;\;\;\;\;\;\;\;\;\;\;\;\;\;
		\text{[by $g(v) = f^*_{i,j}(v)$]}  \\
		& = \frac{  D_{\texttt{prev}(i,j)}}{D_{i,j}} f^*_{w'_{i,j}}(v)
			\;\;\;\;\;\;\;\;\;\;\;\;\;\;\;\;\;\;\;\;\;\;\;\;\;\;\;\;\;\;\;\;\;\;\;\;\;\;\;\;\;\;\;\;\;
		\text{[by $f^*_{i,j}=f^*_{{w'_{i,j}}}/D_{i,j}$]} \\
		& \succeq \frac{  D_{\texttt{prev}(i,j)}}{D_{i,j}} f^*_{w'_{i,j}}(u) \ominus
 \frac{D_{\texttt{prev}(i,j)}}{D_{i,j}} w'_{i,j}(v,u) 
		\;\;\;\;\;\;\;
			\text{[$f^*_{w'_{i,j}}$ is SEPM of $\Gamma^{w'_{i,j}}$]} \\
			& =  D_{\texttt{prev}(i,j)} f^*_{i,j}(u)\ominus
				 \frac{D_{\texttt{prev}(i,j)}}{D_{i,j}} w'_{i,j}(v,u) 
		\;\;\;\;\;\;\;\;\;\; 
		\text{[by $f^*_{i,j}=f^*_{{w'_{i,j}}}/D_{i,j}$]} \\
		& =  D_{\texttt{prev}(i,j)} f^*_{i,j}(u)\ominus
				 D_{\texttt{prev}(i,j)} w_{i,j}(v,u) 
		\;\;\;\;\;\;\;\;\;\;\;
		 \text{[by $w_{i,j}(v,u) = w'_{i,j}(v,u)/D_{i,j}$]} \\
			& \succeq D_{\texttt{prev}(i,j)} f^*_{i,j}(u)\ominus
				 D_{\texttt{prev}(i,j)}\, w_{\texttt{prev}(i,j)}(v,u)\; 
					\text{[by $w_{i,j}<w_{\texttt{prev}(i,j)}$]}  \\
			& = D_{\texttt{prev}(i,j)}  f^*_{i,j}(u)\ominus w'_{\texttt{prev}(i,j)}(v,u) 
		\;\;\;\;\;\;\;\;\;\;\;\;\;\;\;\;\;\;
	\text{[by $D_{\texttt{prev}(i,j)}w_{\texttt{prev}(i,j)} = w'_{\texttt{prev}(i,j)}$]} \\ 
		& \succeq D_{\texttt{prev}(i,j)} g(u)\ominus w'_{\texttt{prev}(i,j)}(v,u) 
		\;\;\;\;\;\;\;\;\;\;\;\;\;\;\;\;\;\;\;\;\;
				\text{[by definition of $g(u)$]} 
\end{align*}
This means that $(v,u)$ is an 
arc compatible with $D_{\texttt{prev}(i,j)}g$ in $\Gamma^{w'_{\texttt{prev}(i,j)}}$.
\end{itemize} 
\item \emph{Case: $v\in V_1$.} The same argument shows that   
$(v,u)\in E$ is compatible with $D_{\texttt{prev}(i,j)}g$ in $\Gamma^{w'_{\texttt{prev}(i,j)}}$, 
but it holds for \emph{all} $u\in\texttt{post}(v)$ in this case.
\end{itemize}
This proves that $D_{\texttt{prev}(i,j)}\, g$ is a SEPM of $\Gamma^{w'_{\texttt{prev}(i,j)}}$.

Since $f^*_{w'_{\texttt{prev}(i,j)}}$ is the \emph{least} SEPM of $\Gamma^{w'_{\texttt{prev}(i,j)}}$, 
then: \[f^*_{w'_{\texttt{prev}(i,j)}}(v)\preceq D_{\texttt{prev}(i,j)}\, g(v), \text{ for every $v\in V$}.\]
Since $f^*_{w'_{\texttt{prev}(i,j)}}=D_{\texttt{prev}(i,j)}\,f^*_{\texttt{prev}(i,j)}$ and  
$g=\min( f^*_{\texttt{prev}(i,j)}, f^*_{i,j})$, then:
\[
D_{\texttt{prev}(i,j)}\,f^*_{\texttt{prev}(i,j)} \preceq D_{\texttt{prev}(i,j)}\, \min( f^*_{\texttt{prev}(i,j)}, f^*_{i,j}).
\]
Whence $f^*_{\texttt{prev}(i,j)} = \min( f^*_{\texttt{prev}(i,j)}, f^*_{i,j})$.

This proves that $f^*_{\texttt{prev}(i,j)}(v) \preceq f^*_{i,j}(v)$ holds for every $v\in V$. 
\item \emph{Fact~1.} 
Next, we prove that if Item~\ref{lemma:correctness:item:2} holds at the $(i,j)$-th scan phase, 
then both Item~\ref{lemma:correctness:item:3} and Item~\ref{lemma:correctness:item:4} hold at the $(i,j)$-th scan phase as well.

\begin{proof}[of Fact~1]
Assume that Item~\ref{lemma:correctness:item:2} holds.
Let us prove Item~\ref{lemma:correctness:item:3} first. 
Since $f^*_{\texttt{prev}(i,j)} \preceq f^*_{i,j}$ holds by Item~\ref{lemma:correctness:item1}, 
and since $\texttt{prev\_}f=f^*_{\texttt{prev}(i,j)}$ holds by hypothesis, 
then $\texttt{prev\_}f(v)\preceq f^*_{i,j}(v)$ holds for every $v\in V$.
Since $w'_{i,j}=D_j\, w_{i,j}$ and $f^*_{w'_{i,j}}=D_j\, f^*_{i,j}$, 
then $D_j\, \texttt{prev\_f}(v)\preceq f^*_{w'_{i,j}}(v)$ holds for every $v\in V$.
Since $w'_{i,j}(e)\in \Z$ for every $e\in E$, then $f^*_{w'_{i,j}}(v)\in \Z$ for every $v\in V$, 
so that $\lceil D_j\, \texttt{prev\_f}(v)\rceil \preceq f^*_{w'_{i,j}}(v)$ holds for every $v\in V$ as well.
This proves Item~\ref{lemma:correctness:item:3}.

We show Item~\ref{lemma:correctness:item:4} now.
Since Item~\ref{lemma:correctness:item:3} holds, 
at line~\ref{algo:solve:l13} it is correct to initialize the starting energy levels of \texttt{Value-Iteration()} to 
$\lceil D_j\, \texttt{prev\_}f(v) \rceil$ for every $v \in V$, 
in order to execute the Value-Iteration on input $\Gamma^{w'_{i,j}}$. 

This implies the following:
\[\texttt{Value-Iteration}(\Gamma^{w'_{i,j}}, \lceil D_j\, \texttt{prev\_}f \rceil) = f^*_{w'_{i,j}}.\]
We know that $\frac{1}{D_j}\, f^*_{w'_{i,j}}=f^*_{i,j}$. 

This proves that Item~\ref{lemma:correctness:item:4} holds and concludes the proof of Fact~1.
\end{proof}

\item\emph{Fact~2.} 
We now prove that Item~\ref{lemma:correctness:item:2} holds at each iteration of line~\ref{algo:solve:l8}.

\begin{proof}[of Fact~2]
The proof proceeds by induction on $(i,j)$.

\emph{Base Case.} Let us consider the first iteration of line~\ref{algo:solve:l8}; 
\ie the iteration indexed by $i=-W$ and $j=1$. 
Recall that, at line~\ref{algo:solve:l2} of Algorithm~\ref{Algo:solve_mpg}, 
the function $f$ is initialized as $f(v)=0$ for every $v\in V$.
Notice that $f$ is really the least SEPM  
$f^*_{-W, 0}$ of $\Gamma_{-W, 0}=\Gamma^{w+W}$, because every 
arc $e\in E$ has a non-negative weight in $\Gamma^{w+W}$, \ie $w_e+W\geq 0$ for every $e\in E$.

Hence, at the first iteration of line~\ref{algo:solve:l8}, the following holds:
\[\texttt{prev\_}f = \textbf{0} = f^*_{-W, 0} = f^*_{\texttt{prev}(-W, 1)} .\]

\emph{Inductive Step.} Let us assume that Item~\ref{lemma:correctness:item:2} 
holds for the $\texttt{prev}(i,j)$-th iteration, and let us prove it for the $(i,j)$-th one. 
Hereafter, let us denote $(i_p, j_p)=\texttt{prev}(i,j)$ for convenience.
Since Item~\ref{lemma:correctness:item:2} holds 
for the $(i_p, j_p)$-th iteration by induction hypothesis, then, by Fact~1, the following 
holds at the $(i_p, j_p)$-th iteration of line~\ref{algo:solve:l13}:
\[ \frac{1}{D_{j_p}}\, \texttt{Value-Iteration}(\Gamma^{w'_{i_p, j_p}}, 
\lceil D_{j_p}\, \texttt{prev\_}f \rceil) = f = f^*_{i_p, j_p}.\]
Thus, at the $(i,j)$-th iteration of line~\ref{algo:solve:l8}: 
\[\texttt{prev\_}f = f = f^*_{i_p, j_p} = f^*_{\texttt{prev}(i,j)}.\]

This concludes the proof of Fact~2.
\end{proof}
\end{itemize}  
At this point, by Fact~1 and Fact~2, Lemma~\ref{lemma:correctness} follows.
\end{proof}

We are now in the position to show that Algorithm~\ref{Algo:solve_mpg} is correct.
\begin{Prop}\label{prop:correctness}
Assume that Algorithm~\ref{Algo:solve_mpg} is invoked on input $\Gamma = ( V, E, w, \langle V_0, V_1\rangle )$  
and, whence, that it returns $( \W_0, \W_1, \nu, \sigma_0 )$ as output.

Then, $\W_0$ and $\W_1$ are the winning sets of Player~0 and Player~1 in $\Gamma$ (respectively), 
$\nu:V\rightarrow S$ is such that $\nu(v)=\val{\Gamma}{v}$ for every $v\in V$, and $\sigma_0:V_0\rightarrow V$ is an 
optimal positional strategy for Player~0 in the \MPG $\Gamma$.
\end{Prop}
\begin{proof}
At the $(i,j)$-th iteration of line~\ref{algo:solve:l15}, 
the following holds by Lemma~\ref{lemma:correctness}: 
\[ \texttt{prev\_}f = f^*_{\texttt{prev}(i,j)}\text{ and } f=f^*_{i,j}. \]
Our aim now is that to apply Theorem~\ref{Thm:transition_opt_values}.
For this, firstly observe that one can safely write $\texttt{prev\_}f = f^*_{i,j-1}$.
In fact, since $F_0=0$ and $F_{s-1}=1$, then: 
\[w_{\texttt{prev}(i, 1)} = w_{i-1,s-1} = w-i = w_{i,0},\;\; \text{ for every }i\in [-W, W].\]
This implies that $w_{\texttt{prev}(i,j)}=w_{i,j-1}$ for every $i\in [-W, W]$ and $j\in [1, s-1]$.

Whence, $\texttt{prev\_f}=f^*_{\texttt{prev}(i,j)}=f^*_{i,j-1}$. 

So, at the $(i,j)$-th iteration of line~\ref{algo:solve:l15}, the following holds for every $v\in V$:
\begin{align*}
\texttt{prev\_}f(v)\neq\top \text{ and } f(v) = \top & 
		\;\text{\emph{ iff }}\; f^*_{i,j-1}(v)\neq\top \text{ and } f^*_{i,j}(v) = \top 
		\;\;  \text{[by Lemma~\ref{lemma:correctness}]} \\ 
	& \;\text{\emph{ iff }}\; v\in \W_0(\Gamma_{i,j-1}) \cap \W_1(\Gamma_{i,j}) 
			\;\;	\text{[by Item~1-2 of Lemma~\ref{Lem:least_energy_prog_measure}]} \\
	& \;\text{\emph{ iff }}\; \val{\Gamma}{v}=i+F_{j-1} 
		\;\;  \text{[by Theorem~\ref{Thm:transition_opt_values}]}
\end{align*}
This implies that, at the $(i,j)$-th iteration of line~\ref{algo:solve:l16}, 
Algorithm~\ref{Algo:solve_mpg} correctly assigns the value $\nu(v)=i+F[j-1]=i+F_{j-1}$ to the vertex $v$.

Since for every vertex $v\in V$ we have $\val{\Gamma}{v}\in S_\Gamma$  
(recall that $S_\Gamma$ admits the following representation 
$S_\Gamma=\left\{i+F_j \suchthat i\in [-W, W),\, j\in [0, s-1]\right\}$), 
then, as soon as Algorithm~\ref{Algo:solve_mpg} halts, 
$\nu(v)=\val{\Gamma}{v}$ correctly holds for every $v\in V$.
In turn, at line~\ref{algo:solve:l18} and at line~\ref{algo:solve:l20}, 
the winning sets $\W_0$ and $\W_1$ are correctly assigned as well.

Now, let us assume that  
$\nu(v)=i+F_{j-1}$ holds at the $(i,j)$-th iteration of line~\ref{algo:solve:l16}, for some $v\in V$.
Then, the following holds on $\texttt{prev\_}w$ at line~\ref{algo:solve:l9}: 
\[ 
	\texttt{prev\_}w = w_{\texttt{prev}(i,j)} = w_{i, j-1} = w-i-F_{j-1} = w-\nu(v) = w-\val{\Gamma}{v}.
\]
Thus, at the $(i,j)$-th iteration of line~\ref{algo:solve:l23}, for every $v\in V_0$ and $u\in \texttt{post}(v)$:
\begin{align*}
\texttt{prev\_}f(v) \succeq \texttt{prev\_}f(u) \ominus \texttt{prev\_} & w(v,u)  \text{\emph{ iff }} 
		f^*_{\texttt{prev}(i,j)}(v) \succeq f^*_{\texttt{prev}(i,j)}(u) \ominus \big(w - \val{\Gamma}{v}\big) \\ 
		& \text{\emph{ iff }}  (v, u) \text{ is compatible with } 
		f^*_{\texttt{prev}(i,j)}\text{ in } \Gamma^{w-\val{\Gamma}{v}} 
\end{align*}
Recall that $f^*_{\texttt{prev}(i,j)}$ is the least SEPM of $\Gamma^{w-\val{\Gamma}{v}}$, 
thus by Theorem~\ref{Thm:pos_opt_strategy} the following implication holds: 
if $(v, u) \text{ is compatible with } f^*_{\texttt{prev}(i,j)}\text{ in } \Gamma^{w-\val{\Gamma}{v}}$, 
then $\sigma_0(v)=u$ is an optimal positional strategy for Player~0, at $v$, in the \MPG $\Gamma$.

This implies that line~\ref{algo:solve:l24} of 
Algorithm~\ref{Algo:solve_mpg} is correct and concludes the proof.
\end{proof}

\subsection{Complexity Analysis}

The present section aims to show that 
Algorithm~\ref{Algo:solve_mpg} always halts in $O(|V|^2|E|\, W)$ time. 
This upper bound is established in the next proposition.
\begin{Prop}\label{prop:complexity}
Algorithm~\ref{Algo:solve_mpg} always halts within 
$O(|V|^2|E|\, W)$ time and it works with $O(|V|)$ space, 
on any input \MPG $\Gamma = (V, E, w, \langle V_0, V_1\rangle)$.
Here, $W=\max_{e\in E}|w_e|$.
\end{Prop}
\begin{proof}

\emph{(Time Complexity of the Init Phase)} 
The initialization of $\W_0, \W_1, \nu, \sigma_0$ (at line~\ref{algo:solve:l1}) and 
that of $f$ (at line~\ref{algo:solve:l2}) takes time $O(|V|)$. 
The initialization of $W$ at line~\ref{algo:solve:l3} takes $O(|E|)$ time.
To conclude, the size $s=|\F_{|V|}|$ of the Farey sequence 
(\ie its total number of terms) can be computed in 
$O(n^{2/3}\log^{1/3}{n})$ time as shown by Pawlewicz and P\u{a}tra\c{s}cu~in~\cite{PaPa09}.
Whence, the Init phase of Algorithm~\ref{Algo:solve_mpg} takes $O(|E|)$ time overall.

\emph{(Time Complexity of the Scan Phases)}
To begin, notice that there are $O(|V|^2 W)$ scan phases overall.
In fact, at line~\ref{algo:solve:l6} the index $i$ goes from $-W$ to $W$, 
while at line~\ref{algo:solve:l7} the index $j$ goes from $0$ to $s-1$ where $s=|\F_{|V|}|=\Theta(|V|^2)$. 
Observe that, at each iteration, it takes $O(|E|)$ time to go from line~\ref{algo:solve:l8} to line~\ref{algo:solve:l12} 
and then from line~\ref{algo:solve:l14} to line~\ref{algo:solve:l25}.
In particular, at line~\ref{algo:solve:l5}, the $j$-th term $F_j$ of the Farey sequence $\F_{|V|}$ can be computed 
in $O(n^{2/3}\log^{4/3}{n})$ time as shown by Pawlewicz and P\u{a}tra\c{s}cu~in~\cite{PaPa09}.

Now, let us denote by $T^{\ref{algo:solve:l13}}_{i,j}$ 
the time taken by the $(i,j)$-th iteration of line~\ref{algo:solve:l13},  
that is the time it takes to execute the Value-Iteration algorithm on input $\Gamma^{w'_{i,j}}$ with 
initial energy levels: $\lceil D_j\, f^*_{\texttt{prev}(i,j)}\rceil$. 
Then, the $(i,j)$-th scan phase always completes within the following time bound: $O(|E|) + T^{\ref{algo:solve:l13}}_{i,j}$.

We now focus on $T^{\ref{algo:solve:l13}}_{i,j}$  
and argue that the (aggregate) total cost $\sum_{i,j} T_{i,j}^{\ref{algo:solve:l13}}$ 
of executing the Value-Iteration algorithm for \EG{s} at line~\ref{algo:solve:l13} 
(throughout all scan phases) is only $O(|V|^2|E|\,W)$. 
Stated otherwise, we aim to show that the \emph{amortized cost} of executing the $(i,j)$-th scan phase is only $O(|E|)$.

Recall that the Value-Iteration algorithm for \EG{s} consists, as a first step, 
into an \emph{initialization} (which takes $O(|E|)$ time)
and, then, in the continuous iteration of the following two operations:
(1) the application of the \emph{lifting} operator $\delta(f, v)$ (which takes $O(|\texttt{post}(v)|)$ time) 
in order to resolve the inconsistency of $f$ in $v$, 
where $f(v)$ represents the current energy level and $v\in V$ is any vertex at which $f$ is inconsistent;
and (2) the \emph{update} of the list $\L$ (which takes $O(|\texttt{pre}(v)|)$ time), 
in order to keep track of all the vertices that witness an inconsistency. 
Recall that $\L$ contains no duplicates. 

At this point, since at the $(i,j)$-th iteration of line~\ref{algo:solve:l13} the Value-Iteration is executed on 
input $\Gamma^{w'_{i,j}}$, then a scaling factor on the maximum absolute weight $W$ must be taken into account.
Indeed, it holds that: 
\begin{align*}
W' & := \max\Big\{|w'_{i,j}(e)| \;\Big|\; e\in E,\; i\in[-W, W],\; j\in[0, s-1]\Big\} = O( |V|\,W ).
\end{align*}
\emph{Remark.} Actually, since $w'_{i,j}:=D_j(w-i)-N_j$ (where $N_j/D_j=F_j\in\F_{|V|}$), then the scaling factor $D_j$ 
\emph{changes} from iteration to iteration. Still, $D_j\leq |V|$ holds for every $j$.   

At each application of the lifting operator 
$\delta(f,v)$ the energy level $f(v)$ increases by at least 
one unit with respect to the scaled-up maximum absolute weight $W'$. 
Stated otherwise, at each application of $\delta(f,v)$, 
the energy level $f(v)$ increases by at least $1/|V|$ units with respect to the original weight $W$.

Throughout the whole computation,
the rational scalings of the energy levels never decrease by Lemma~\ref{lemma:correctness}:
in fact, at the $(i,j)$-th scan phase, Algorithm~\ref{Algo:solve_mpg} executes the Value-Iteration 
with initial energy levels: $\lceil D_j\, f^*_{\texttt{prev}(i,j)}\rceil$.
Whence, at line~\ref{algo:solve:l13}, the $(i,j)$-th execution of the Value-Iteration starts from the
(carefully scaled-up) energy levels of the $\texttt{prev}(i,j)$-th execution; 
roughly speaking, no energy gets ever lost during this process.
Then, by Lemma~\ref{Lem:lepm_equals_mincredit}, 
each energy level $f(v)$ can be lifted-up at most $|V|\,W' = O( |V|^2\, W )$ times.

The above observations imply that 
the (aggregate) total cost of executing the Value-Iteration at line~\ref{algo:solve:l13} 
(throughout all scan phases) can be bounded as follows: 
\begin{align*}
\sum_{\substack{-W \leq i \leq W \\ 1\leq j\leq s-1}} 
	T^{\ref{algo:solve:l13}}_{i,j} & = 
			 \left(\sum_{\substack{-W \leq i \leq W \\ 1\leq j\leq s-1}} 
			\underbrace{O(|E|)}_{\text{\emph{init cost}}}\right)  + \left(
			\sum_{v\in V}O\big(\underbrace{|\texttt{post(v)}|}_{\text{\emph{lifting $\delta$}}} + 
			\underbrace{|\texttt{pre}(v)|}_{\text{\emph{update $\L$ }}} \big) 
			\underbrace{O(|V|\, W')}_{Lemma~\ref{Lem:lepm_equals_mincredit}} \right)\\ 
			    & = O(|V|^2 |E|\, W) + O(|V|^2 W)\, \sum_{v\in V}  
				O\big(|\texttt{post(v)}| + |\texttt{pre}(v)| \big) \\
			& = O(|V|^2 |E|\, W) 
\end{align*}
Whence, Algorithm~\ref{Algo:solve_mpg} always halts within the following time bound: 
\begin{align*}
\textsc{Time}\Big(\texttt{solve\_MPG}\big(\Gamma\big)\Big) 
			=\sum_{\substack{-W \leq i \leq W \\ 1\leq j\leq s-1}} \left( O(E) + T^{\ref{algo:solve:l13}}_{i,j} \right)  
				 &  = O(|V|^2 |E|\, W). 
\end{align*}
This concludes the proof of the time complexity bound. 

We now turn our attention to the space complexity.

\emph{(Space Complexity)}
First of all, although the Farey sequence $\F_{|V|}$ has $|\F_{|V|}|=\Theta(|V|^2)$ many elements, 
still, Algorithm~\ref{Algo:solve_mpg} works fine assuming that every next element of the sequence 
is generated \emph{on the fly} at line~\ref{algo:solve:l5}. 
This computation can be computed in 
$O(|V|^{2/3}\log^{4/3}{|V|})$ \emph{sub-linear} time and space as shown by  
Pawlewicz and P\u{a}tra\c{s}cu~\cite{PaPa09}. 
Secondly, given $i$ and $j$, it is not necessary to actually store all weights 
$w'_{i,j}(e):=D_j(w(e)-i)-N_j$ for every $e\in E$, 
as one can compute them on the fly provided that $N_j$, $D_j$, $w$ and $e$ are given. 
Finally, Algorithm~\ref{Algo:solve_mpg} needs to 
store in memory the two SEPMs $f$ and $\texttt{old\_}f$, 
but this requires only $O(|V|)$ space. Finally, at line~\ref{algo:solve:l13}, 
the Value-Iteration algorithm employs only $O(|V|)$ space. 
In fact the list $\L$, which it maintains in order to keep 
track of inconsistencies, doesn't contain duplicate vertices and, 
therefore, its length is at most $|\L|\leq|V|$.
These facts imply altogether that Algorithm~\ref{Algo:solve_mpg} works with $O(|V|)$ space.
\end{proof}

\section{Conclusions}\label{sect:conclusions}
In this work we proved an $O(|V|^2 |E|\, W)$ pseudo-polynomial time upper bound 
for the Value Problem and Optimal Strategy Synthesis in Mean Payoff Games.
The result was achieved by providing a suitable description of 
values and positional strategies in terms of reweighted Energy Games  
and Small Energy-Progress Measures. 

On this way we ask whether further improvements are not too far away. 

\paragraph{Acknowledgements.}
This work was supported by \emph{Department of Computer Science, 
University of Verona, Verona, Italy},
under PhD grant “Computational Mathematics and Biology”, 
on a co-tutelle agreement 
with \emph{LIGM, Universit\'e Paris-Est in Marne-la-Vall\'ee, Paris, France}.

\bibliographystyle{spmpsci}
\bibliography{biblio}
\end{document}